\documentclass[12pt]{article}
\usepackage{amsfonts,epsfig,latexsym,amscd,amsmath,mathrsfs,color,xypic}

\usepackage{amssymb,amsthm,psfrag}

\newcommand{\RR}{{\mathbb{R}}}
\newcommand{\CC}{{\mathbb{C}}}
\newcommand{\ZZ}{{\mathbb{Z}}}

\newcommand{\pa}{\partial}
\newcommand{\ii}{{\rm i}}
\newcommand{\dd}{{\rm d}}

\newcommand{\sfrac}[2]{{\textstyle\frac{#1}{#2}}}
\newcommand{\Tr}{\mathrm{Tr}}

\newcommand{\g}{\mathfrak{g}}
\newcommand{\gggg}{e}
\newcommand{\su}{\mathfrak{su}}
\newcommand{\R}{{\mathbb{R}}}
\newcommand{\Z}{{\mathbb{Z}}}
\newcommand{\N}{{\mathbb{N}}}

\newcommand{\E}{{\mathbb{E}}}

\newcommand{\SU}{{\mathrm{SU}}}

\newcommand{\xvec}{\mathbf{x}}
\newcommand{\ivec}{\mathbf{i}}
\newcommand{\jvec}{\mathbf{j}}
\newcommand{\kvec}{\mathbf{k}}
\newcommand{\cvec}{\mathbf{c}}
\newcommand{\nvec}{\mathbf{n}}
\newcommand{\Xvec}{\mathbf{X}}
\newcommand{\omvec}{\mbox{\boldmath $\omega$}}
\newcommand{\sigvec}{\mbox{\boldmath $\sigma$}}
\newcommand{\beq}{\begin{equation}}
\newcommand{\eeq}{\end{equation}}
\newcommand{\bea}{\begin{eqnarray}}
\newcommand{\eea}{\end{eqnarray}}
\newcommand{\ben}{\begin{eqnarray*}}
\newcommand{\een}{\end{eqnarray*}}
\newcommand{\bem}{\begin{enumerate}}
\newcommand{\eem}{\end{enumerate}}
\newcommand{\Id}{{\mathrm{Id}}}
\newcommand{\ra}{\rightarrow}

\newcommand{\wt}{\widetilde}

\newcommand{\ol}{\overline}
\newcommand{\less}{\backslash}

\newcommand{\ignore}[1]{}

\theoremstyle{plain}
\newtheorem{theorem}{Theorem}

\newtheorem{proposition}[theorem]{Proposition}
\newtheorem{prop}[theorem]{Proposition}
\theoremstyle{definition}

\newcommand{\news}{\setcounter{equation}{0}}

\begin{document}

\title{A point particle model of lightly bound skyrmions}
\author{Mike Gillard${}^2$, Derek Harland${}^1$, Elliot Kirk${}^3$, Ben Maybee${}^1$ \\ and Martin Speight${}^1$\\
\small ${}^1$ School of Mathematics, University of Leeds, Leeds LS2 9JT, UK \\
\small ${}^2$ The Wolfson School of Engineering, University of Loughborough, Loughborough LE11 3TU, UK\\
\small ${}^3$ Department of Mathematical Sciences, University of Durham, Durham DH1 3LE, UK}

\normalsize
\date{}

\maketitle

\begin{abstract}
A simple model of the dynamics of lightly bound skyrmions is developed in which skyrmions are replaced by point particles, each
carrying an internal orientation.
The model accounts well for the static energy
minimizers of baryon number $1\leq B\leq 8$ obtained by numerical simulation of the full field theory. For $9\leq B\leq 23$, a large number of static solutions of the point particle model are found, all closely resembling size $B$ subsets of a face centred cubic lattice, with
the particle orientations dictated by a simple colouring rule. Rigid body quantization of these solutions is performed, and the spin and isospin of the corresponding ground states extracted. As part of the quantization scheme, an algorithm to compute the symmetry group of an oriented point cloud, and to determine its corresponding Finkelstein-Rubinstein constraints, is devised. 
\end{abstract}

\section{Introduction}
\label{sec:1}
\news

The Skyrme model is an effective theory of nuclear physics in which nucleons emerge as topological solitons in a field whose small amplitude travelling waves represent pions. It thus provides a unified treatment of both nucleons and the mesons which, in the Yukawa picture, are responsible for the strong nuclear forces between them. While the Skyrme model has been superceded as a fundamental model of strong interactions by QCD, interest in the model revived once it was recognized to be a possible low energy reduction of QCD in the limit of large $N_c$ (number of colours)
\cite{wit-sky1,wit-sky2}, and much work has been conducted to extract phenomenological predictions about nuclei from standard versions of the model
\cite{leemansch,manmanwoo,lauman,feilauman}. Many of these predictions are in good qualitative agreement with experiment, and recent improvements in skyrmion quantization schemes offer hope of significant further improvement to come
\cite{hal,halkinman}. 

One area in which standard versions of the model perform poorly, however, is that of nuclear binding energies: typically, classical skyrmions are much more tightly bound than the nuclei they are meant to represent (by a factor of 15 or so). In recent years, no fewer than three variants of the model have been proposed which seek to remedy this problem. In each case, the model is, by design, a small perturbation of a Skyrme model in which the binding energies vanish exactly. Perhaps the most radical proposal,
due to Sutcliffe and
motivated by holography, couples the Skyrme field to an infinite tower of vector mesons \cite{sut-holographic}. Small but nonvanishing binding energies are (conjecturally) introduced by truncating this infinite tower at some high but finite level. This proposal, while elegant, has so far not been amenable to detailed analysis. A second proposal, due to Adam, Sanchez-Guillen and Wereszczynzki, starts with a model which is invariant under volume preserving diffeomorphisms of space, then perturbs it by mixing with a small fraction of the conventional Skyrme energy \cite{adasanwer}. Skyrmions in this model have the attractive feature of being somewhat akin to liquid drops. However, the large (in fact, infinite dimensional) symmetry group of the unperturbed model is extremely problematic for numerical simulations, and the shapes and symmetries of classical skyrmions, even for rather low baryon number ($B\geq 3$) are, so far, not known in this model in the regime of realistically small binding energy \cite{gilharspe}. 

In this paper we will study the third (and arguably least radical) proposal, originally due to one of us \cite{har1}. This amounts to making a nonstandard choice of potential term in the standard Skyrme lagrangian and, more importantly, radically shifting the weighting of the derivative terms from the quadratic to the quartic. The resulting model is still amenable to numerical simulation, but its classical solutions are quite different from conventional skyrmions: the lowest energy Skyrme field of baryon number $B$ now resembles a loosely bound collection of $B$ spherically symmetric unit skyrmions, rather than a tightly bound object in which the skyrmions have merged and lost their individual identities. In the terminology of \cite{salsut}, which studied a $(2+1)$ dimensional analogue of the model, skyrmions in this lightly bound Skyrme model prefer to hold themselves aloof from one another. Numerical analysis reveals \cite{gilharspe} that they also prefer to arrange themselves on the vertices of a face centred cubic spatial lattice, with internal orientations dictated by their lattice position. This suggests that, unlike conventional
skyrmions, lightly bound skyrmions can be modelled as point particles, each carrying an internal orientation, interacting with one another through some pairwise interaction potential whose minimum encourages them to sit at a fixed separation with their internal orientations correlated. The aim of this paper is to derive such a simple point particle model, compare its predictions with numerical simulations of the full field theory, and use it to extract, via rigid body quantization, phenomenological predictions about nuclei with baryon number $2\leq B\leq 23$. A similar programme (minus quantization) for the $(2+1)$ dimensional analogue model was completed in \cite{salsut}. 

As we shall see, the point particle model accounts almost flawlessly for static skyrmions with $1\leq B\leq 8$, where comparison with simulations of the full field theory is available. For $B\geq 9$, it predicts a rapid proliferation of nearly degenerate skyrmions as $B$ grows, all rather close to size $B$ subsets of the face centred cubic lattice. In comparison with conventional skyrmions, these typically have rather little symmetry, and anisotropic mass distribution. Determining the symmetries of these configurations is an interesting and important task, nonetheless, as they determine the Finkelstein-Rubinstein constraints on quantization. Usually, symmetries of skyrmions are determined by ad hoc means: one looks at suitable pictures of the skyrmion, predicts a symmetry by eye, then checks it by operating on the numerical data. By contrast, we will develop an algorithm which automatically computes the symmetry group of any point particle configuration. This allows us to completely automate the rigid body quantization scheme. The result is, as a phenomenological model of nuclei, moderately successful: rigid body ground states plausibly account for the lightest nucleus of baryon number $B$ for 12 of the 23 values considered. Presumably this can be improved by replacing rigid body quantization by something more sophisticated.

The rest of the paper is structured as follows. In section \ref{sec:lbdsm} we review the lightly bound Skyrme model, focussing on its spin-isospin symmetry and associated inertia tensors. In section \ref{sec:ppm} we introduce the point particle model, then in section
\ref{sec:numres} we describe a numerical scheme to find its energy minimizers, and present the results of this scheme. In section
\ref{sec:rbq} we formulate the rigid body quantization of our classical energy minimizers, focussing particularly on the Finkelstein-Rubinstein constraints. Some concluding remarks and possible future directions of development are presented in section 
\ref{sec:pontificate}.

\section{The lightly bound Skyrme model}
\label{sec:lbdsm}
\news

The field theory of interest is defined as follows. There is a single 
Skyrme field $U:\RR^{3,1}\to\mathrm{SU}(2)$, required to satisfy the boundary condition $U(t,\xvec)=1$ as $|\xvec|\ra\infty$ for all $t$. Such a field, if smooth, has at each $t$, a well-defined integer valued topological charge
\begin{equation}
\label{baryon number}
 B = -\frac{1}{24\pi^2}\int_{\RR^3}\epsilon_{ijk}\Tr(R_iR_jR_k) \dd^3 x,
\end{equation}
the topological degree of the map $U(t,\cdot):\R^3\cup\{\infty\}\ra \mathrm{SU}(2)\cong S^3$. Since the field is smooth, $B(t)$ is smooth and integer valued, hence automatically conserved. Physically it is interpreted as the baryon number of the field $U$. The
right invariant current associated with $U$ is $R_\mu=(\pa_\mu U)U^\dagger$, in terms of which the lagrangian density is
\begin{multline}
\label{Skyrme lagrangian}
 \mathcal{L} = \frac{F_\pi^2}{16\hbar} \Tr(R_\mu R^\mu) + \frac{\hbar}{32\gggg^2} \Tr([R_\mu,R_\nu],[R^\mu,R^\nu]) \\
 - \frac{F_\pi^2 m_\pi^2}{8\hbar^3}\Tr(1-U) - \frac{F_\pi^4\gggg^2\alpha}{32(1-\alpha)^2} (\sfrac12\Tr(1-U))^4 .
\end{multline}
Here $F_\pi$ is the pion decay constant, $m_\pi$ the pion mass, and $\gggg>0$, $0\leq\alpha<1$ are dimensionless parameters.  In \cite{gilharspe} the following values were chosen for these parameters so that classical binding energies in the model are comparable with experimentally-measured nuclear binding energies:\footnote{The value for $F_\pi$ recorded here corrects a typographical error in \cite{gilharspe}}
\begin{equation}
 F_\pi = 36.1\,\mathrm{MeV},\quad m_\pi = 303\,\mathrm{MeV},\quad \gggg = 3.76,\quad \alpha=0.95.
\end{equation}
There is certainly room for improvement in this calibration: for example, obtaining the correct pion mass was not a priority in \cite{gilharspe}, and we expect that a more thorough analysis could result in a parameter set for which $m_\pi$ is closer to its experimental value of 137MeV.  However, the aim in the present paper is not to fine-tune the parameters, but rather to study qualitative properties of static solutions, which we expect to be insensitive to details of the calibration.

It will be convenient to use $F_\pi/4\gggg\sqrt{1-\alpha}$ as a unit of energy and $2\sqrt{1-\alpha}/F_\pi \gggg$ as a unit of length; in these units the lagrangian takes the form $L=T-V$, where
\begin{align}
T &= \int_{\RR^3}\Big[ -\frac12(1-\alpha)\Tr(R_0R_0) - \frac18\Tr([R_0,R_i][R_0,R_i]) \Big]\dd^3 x, \\
\nonumber
V &= \int_{\RR^3} \Big[ (1-\alpha)\left( -\frac12\Tr(R_iR_i) + m^2\Tr(1-U) \right) \\
& \qquad\qquad- \frac{1}{16}\Tr([R_i,R_j][R_i,R_j]) + \alpha(\sfrac12\Tr(1-U))^4 \Big] \dd^3 x,
\label{Skyrme energy}
\end{align}
and $m:=(2m_\pi\sqrt{1-\alpha}/F_\pi \gggg)$. In the parameter set given above, $m=1.00$. Note that when $\alpha=0$, $L$ is the lagrangian of the conventional Skyrme model with pion mass, while for $\alpha=1$ this is a completely unbound model \cite{har1}: there is a topological energy bound of the form $V\geq \mathrm{const}|B|$, but this is attained only when $|B|\leq 1$. 

The first approximation to a nucleus containing $B$ nucleons is a static Skyrme field $U:\RR^3\to\mathrm{SU}(2)$ of degree
$B$ which minimizes the potential energy $V$.  Thus it is important to identify static classical energy minimizers.  These are referred to as skyrmions.
A better approximation to a nucleus is obtained by allowing solitons to carry spin and isospin.  The lagrangian is invariant under a left action of the group $G:=\mathrm{SU}(2)_I\times\mathrm{SU}(2)_J$, defined by
\beq
[(g,h)\cdot U](t,\mathbf{x}):= gU(t,h^{-1}\mathbf{x}h) g^{-1}
\eeq
where we have identified physical space $\R^3$ with the Lie algebra $\su(2)$ via $\xvec \cong ix^j\sigma_j$, $\sigma_1,\sigma_2,\sigma_3$ being the Pauli matrices, to define the action of $h$ on $\xvec$. Equivalently,
\beq
 [(g,h)\cdot U](t,\mathbf{x}):= gU(t,R(h)^{-1}\mathbf{x}) g^{-1}
\eeq
where $R(h)$ is the $SO(3)$ matrix with entries
\beq
R(h)_{ij}=\frac12\Tr
 (h\sigma_ih^{-1}\sigma_j). 
\eeq
The conserved quantities associated with these symmetries are isospin and spin.  We refer to transformations $g\in\mathrm{SU}(2)_I$ as isorotations, in analogy with rotations $h\in\mathrm{SU}(2)_J$.

Every $\omega\in\g:=\su(2)_I\oplus\su(2)_J$ defines a one-parameter subgroup
$\{\exp(t\omega)\: :\: t\in\R\}$ of $G$ isomorphic to $S^1$, whose action on a static 
skyrmion $U$ generates a rigidly isorotating and rotating skyrmion, $U_\omega=\exp(t\omega)\cdot U$, of constant kinetic energy
$T[U_\omega]$. The mapping $\omega\mapsto T[U_\omega]$ is a quadratic form on $\g$, and hence defines a unique symmetric bilinear form
$\Lambda:\g\times\g\ra\R$ called the \emph{inertia tensor} of the skyrmion $U$. 
By its definition, $\Lambda$ vanishes on the subspace of $\g$ tangent to the
isotropy group $G^U$ of $U$ (that is, the subgroup $G^U:=\{(g,h)\in G\: :\: (g,h)\cdot U=U\}<G$ which leaves $U$ unchanged).
If $G^U$ is discrete, as is the case for all the skyrmions studied in this paper except when $B=1$,
then $\Lambda$ is a positive bilinear form, and thus defines a left invariant Riemannian metric on $G$. 
In order to identify spin and isospin quantum numbers of skyrmions corresponding to those of nuclei, isorotations and rotations needed to be treated quantum mechanically rather classically.  The inertia tensor plays an important role in the simplest quantization scheme, known as rigid body quantization, which will be reviewed in section \ref{sec:rbq}, and amounts to quantizing geodesic motion on $(G,\Lambda)$, subject to certain symmetry constraints required to give skyrmions fermionic exchange statistics. Clearly, by choosing a basis for $\su(2)$, we obtain a basis for $\g$ which can be used to represent $\Lambda$ as a real symmetric $6\times 6$ matrix. We shall consistently represent inertia tensors
in this way, having chosen the basis $[-\sfrac\ii2\sigma_1,-\sfrac\ii2\sigma_2,-\sfrac\ii2\sigma_3]$ for $\su(2)$.

\section{The point particle model}
\label{sec:3}\label{sec:ppm}
\news

Extensive numerical simulations reported in \cite{gilharspe} showed that skyrmions in the lightly bound Skyrme model with $B>0$ invariably resemble collections of $B$ particles.  Encouraged by this observation, we have developed a point particle model in which a Skyrme field $U$ with baryon number $B$ is replaced by $B$ oriented point particles in $\RR^3$.

To explain how the model is derived, we begin by recalling the structure of the simplest skyrmion, which has $B=1$, and is of ``hedgehog'' form
\begin{equation}
\label{hedgehog ansatz}
U_H(\mathbf{x}) = \exp(f(r)\ii\sigma_j x_j/r),
\end{equation}
with $f(r)$ a real function satisfying $f(0)=\pi$, $f(r)\to0$ as $r\to\infty$, and $r=|\xvec|$.  The profile function
is determined by solving (numerically) the Euler-Lagrange equation for $V$ restricted to fields of hedgehog form, a certain nonlinear second order ODE for $f$. One finds that $U_H$ has total energy $M_H:=V[U_H]\approx87.49$, and its energy density is monotonically decreasing with $r$ and concentrated around the origin.
The  1-skyrmion has a high degree of symmetry: if $g\in\mathrm{SU}(2)$ then
\[
gU_H(R(g)^{-1}\mathbf{x})g^{-1} = U_H(\mathbf{x}).
\]
In other words, $G^{U_H}$ is the diagonal subgroup of $\mathrm{SU}(2)_I\times\mathrm{SU}(2)_J$.

This basic skyrmion can be moved and rotated using symmetries of the model.  A 1-skyrmion with position $\mathbf{x}_0\in\RR^3$ and orientation $q_0\in\mathrm{SU}(2)$ is given by
\begin{equation}
\label{oriented hedgehog}
U(\mathbf{x};\mathbf{x}_0,q_0) = U_H(R(q_0)(\mathbf{x}-\mathbf{x}_0)).
\end{equation}
The energy-minimizers with $2\leq B\leq 8$ resemble superpositions of fields of this type \cite{gilharspe}.  More precisely, their energy densities are concentrated at $B$ well-separated points $\mathbf{x}_1,\ldots,\mathbf{x}_B$, and near each such point $\mathbf{x}_a$ the field $U$ is approximately of the above form for some $q_a$.  These positions and orientations are the basic degrees of freedom in our point particle model, and will be allowed to depend on time $t$.  The lagrangian for this point particle model takes the form
\begin{equation}
\label{point particle lagrangian}
L_{pp}=\sum_{a=1}^B \left(\frac12M|\dot{\mathbf{x}}_a|^2 + \frac12L |\dot{q}_a|^2\right) - BM - V(\mathbf{x}_1,\ldots,\mathbf{x}_B,q_1,\ldots,q_B),
\end{equation}
where $|\dot{q}|^2:=\frac12\Tr(\dot{q}\dot{q}^\dagger)$ and 
\begin{equation}
\label{point particle potential}
V(\mathbf{x}_1,\ldots,\mathbf{x}_B,q_1,\ldots,q_B) = \sum_{1\leq a<b\leq B} V_{int}(\mathbf{x}_a,q_a,\mathbf{x}_b,q_b),
\end{equation}
is an interaction potential.

The terms involving time derivatives of $\mathbf{x}_a$ and $q_a$ represent the kinetic energy of a moving skyrmion.  Their coefficients could be deduced from the Skyrme model.  It is known that the 1-skyrmion has inertia tensor
\[ \Lambda_H=L_H \begin{pmatrix} \mathrm{Id}_3 & -\mathrm{Id}_3 \\ -\mathrm{Id}_3 & \mathrm{Id}_3 \end{pmatrix}, \]
where
\[ L_H = \frac{16\pi}{3} \int_0^\infty \sin^2f\big((1-\alpha)r^2+r^2(f')^2+\sin^2f\big)\dd r\approx 53.49. \]
From this it follows that the kinetic energy of a rigidly rotating skyrmion should take the form $\frac12L_H |\dot{q}_0|^2$, suggesting that $L=L_H$ in the lagrangian \eqref{point particle lagrangian}.  Similarly, the kinetic energy of a 1-skyrmion moving with velocity $\dot{\mathbf{x}}_0$ is $\frac12M_H|\dot{\mathbf{x}}_0|^2$, where $M_H\approx87.49$ is the potential energy of a static 1-skyrmion.  This suggests choosing $M=M_H$ in the lagrangian.  However, we have chosen to fix the coefficients by an alternative phenomenological method that will be explained in the next section.

\subsection{Symmetries of the interaction potential}

The point particle model inherits an action of $G=\mathrm{SU}(2)_I\times\mathrm{SU}(2)_J$ from the Skyrme model.  The action of $(g,h)\in G$ on the field $U(\mathbf{x};\mathbf{x}_0,q_0)$ defined in equation \eqref{oriented hedgehog} is
\begin{align*}
U(\mathbf{x};\mathbf{x}_0,q_0) &\mapsto gU(R(h)^{-1}\mathbf{x};\mathbf{x}_0,q_0)g^{-1}\\
& = gU_H(R(q_0)(R(h)^{-1}\mathbf{x}-\mathbf{x}_0))g^{-1} \\ 
& = U_H(R(g)R(q_0)R(h)^{-1}(\mathbf{x}-R(h)\mathbf{x}_0)) \\ 
& = U(\mathbf{x};R(h)\mathbf{x}_0,gq_0h^{-1}).
\end{align*}
Therefore the action of $(g,h)$ on a point particle configuration is
\[ (\mathbf{x}_a,q_a)\mapsto (R(h)\mathbf{x}_a,gq_ah^{-1}),\qquad a=1,\ldots,B. \]
The point particle lagrangian should be invariant under these transformations, and under translations $\mathbf{x}_a\mapsto\mathbf{x}_a+\mathbf{c}$ for $\mathbf{c}\in\RR^3$.  It should be invariant under changes of the signs of any of the $q_a$, because $U(\mathbf{x};\mathbf{x}_0,-q_0)=U(\mathbf{x};\mathbf{x}_0,q_0)$.  It should also be invariant under permutations of the particles, because configurations of particles that are the same up to a re-ordering describe the same Skyrme field. Finally, the Skyrme
model is invariant under the inversion
\[
U(\mathbf{x})\mapsto U(-\mathbf{x})^{\dagger},
\]
which is equivalent, for a field of the form \eqref{oriented hedgehog}, to $(\xvec_0,q_0)\mapsto(-\xvec,q_0)$. Hence, our point
particle lagrangian should be invariant under
\beq\label{inv_sym}
(\xvec_a,q_a)\mapsto (-\xvec_a,q_a).
\eeq

The kinetic terms in \eqref{point particle lagrangian} obviously have these symmetries.  Demanding that the potential \eqref{point particle potential} is also invariant imposes constraints on the function $V_{int}(\mathbf{x}_1,q_1,\mathbf{x}_2,q_2)$ which we now describe.

Translation symmetry implies that $V_{int}(\mathbf{x}_1,q_1,\mathbf{x}_2,q_2)$ depends on the positions of the skyrmions only through their relative position $\mathbf{X}:=\mathbf{x}_1-\mathbf{x}_2$.  Isorotation symmetry implies that it depends on $q_1,q_2$ only through the isorotation-invariant combination $Q=q_1^{-1}q_2$.  Thus
\[ V_{int}(\mathbf{x}_1,q_1,\mathbf{x}_2,q_2) = V_{red}(\mathbf{X},Q), \]
for some function $V_{red}$ on $\R^3\less\{0\}\times \mathrm{SU}(2)$. Invariance under $q_1\mapsto-q_1$ implies
\begin{equation}
\label{sym0}
V_{red}(\mathbf{X},-Q)=V_{red}(\mathbf{X},Q),
\end{equation}
while rotational symmetry demands that
\begin{equation}
\label{sym1}
V_{red}(R(h)\mathbf{X},hQh^{-1}) = V_{red}(\mathbf{X},Q)\quad\forall h\in\mathrm{SU}(2)_J.
\end{equation}
A permutation $(\mathbf{x}_1,q_1,\mathbf{x}_2,q_2)\mapsto(\mathbf{x}_2,q_2,\mathbf{x}_1,q_1)$ changes the sign of $\mathbf{X}$ and inverts $Q$, so permutation invariance implies that
\begin{equation}
\label{sym2}
V_{red}(-\mathbf{X},Q^{-1}) = V(\mathbf{X},Q).
\end{equation}
Finally, symmetry under inversion \eqref{inv_sym}, implies
\begin{equation}
\label{sym3}
V_{red}(-\mathbf{X},Q) = V_{red}(\mathbf{X},Q).
\end{equation}

To proceed further, it is helpful to think of $V_{red}$ as a one-parameter family of
real functions $V_\rho$ on $S^2\times\SU(2)$, parametrized by $\rho:=|\Xvec|\in (0,\infty)$. We may expand each such function in a convenient basis for $L^2(S^2\times \SU(2))$, for example, the basis of eigenfunctions of the Laplacian. A natural truncation to finite dimensions is obtained by keeping only eigenfunctions up to a fixed finite eigenvalue. The effect of this truncation is to exclude from $V_{red}$ terms with fast orientation dependence. This motivates the following definition: for each $\lambda$ in the spectrum of $\Delta_{S^2\times S^3}$, let
$E_\lambda$ denote the corresponding eigenspace, and for any $\mu\geq0$,
\beq
F_\mu=\bigoplus_{\lambda\leq\mu}E_\lambda.
\eeq
Let $C^\infty_\mu$ denote the space of smooth functions on $V:\R^3\less\{0\}\times \SU(2)\ra\R$ such that $V_\rho\in F_\mu$ for
all $\rho$. 

\begin{prop}\label{prop1} Let $\mu\in[0,20)$ and $V$ be a function in $C^\infty_\mu$ invariant under the symmetries
\eqref{sym0}-\eqref{sym3}. Then there exist functions $V_i:(0,\infty)\ra\R$, $i=0,1,2$, such that
\beq\label{Vredansatz}
V(\mathbf{X},Q) = V_0(|\mathbf{X}|) + V_1(|\mathbf{X}|)\Tr(R(Q)) + V_2(|\mathbf{X}|) \frac{\mathbf{X}\cdot R(Q)\mathbf{X}}{|\mathbf{X}|^2}.
\eeq
\end{prop}

\begin{proof}
Recall that the eigenvalues of the Laplacian on $S^n$ are $\lambda_d^{(n)}=d(d+n-1)$, $d=0,1,2,\ldots$,
and the corresponding eigenspaces, $\E_d^{(n)}$, are spanned by (the restrictions to $S^n\subset\R^{n+1}$ of) harmonic
homogeneous polynomials in $\R^{n+1}$ of degree $d$ \cite{bergaumaz}. It follows that the eigenvalues of $\Delta_{S^2\times S^3}$ are
$\lambda_d^{(2)}+\lambda_{d'}^{(3)}$ with eigenspaces $\E_{d}^{(2)}\otimes\E_{d'}^{(3)}$. By \eqref{sym0}, \eqref{sym3}, $V$ is
invariant under both $\Xvec\mapsto-\Xvec$ and $Q\mapsto-Q$, so we may restrict $d$ and $d'$ to only even values (homogeneous polynomials of odd degree are parity odd). Further, since $V\in C^\infty_\mu$ with $\mu<20$, each restriction $V_\rho$ lies in
\beq\label{wtf1}
E_0\oplus E_6\oplus E_8\oplus E_{14}=(\E_0^{(2)}\otimes\E_0^{(3)})\oplus(\E_2^{(2)}\otimes\E_0^{(3)})\oplus
(\E_0^{(2)}\otimes\E_2^{(3)})\oplus(\E_2^{(2)}\otimes\E_2^{(3)}).
\eeq
Now $\SU(2)$ acts on both $\E^{(2)}_d$ (by rotations of $S^2$) and $\E^{(3)}_{d'}$ (by conjugation on $\SU(2)$), and, by\eqref{sym1}, each $V_\rho$ is invariant under the combined action. In fact $\E^{(2)}_d\cong\R^{2d+1}$ and carries the irreducible spin $d$ representation of $\SU(2)$, while $\E^{(3)}_{d'}\cong\R^{d'+1}\otimes\R^{d'+1}$ where, for $d'=2\ell$, $\R^{d'+1}$ carries the irreducible spin $\ell$ representation of $\SU(2)$. In particular, $\E_0^{(3)}=\R$, on which $\SU(2)$ acts trivially, and
$\E_0^{(3)}$ decomposes into irreducible representations as
\beq
\E_0^{(3)}=\R\oplus\R^3\oplus\R^5.
\eeq
Now the tensor product $\R^{2d+1}\otimes\R^{2\ell+1}$ contains no trivial subrepresentation if $d\neq\ell$, and exactly one if
$d=\ell$. Hence, of the summands in \eqref{wtf1}, $E_0$, $E_8$ and $E_{14}$ each contain a one-dimensional subspace on which $\SU(2)$ acts trivially (while $E_6$ does not) and, by \eqref{sym1}, $V_\rho$ lies in the three-dimensional space spanned by these. 
Clearly $E_0^{triv}=E_0$ which is spanned by the constant function $(\Xvec,Q)\mapsto 1$. Consider the functions
\beq
(\Xvec,Q)\mapsto \Tr(Q),\qquad (\Xvec,Q)\mapsto \Xvec\cdot R(Q)\Xvec-\frac12\Tr R(Q)|\Xvec|^2.
\eeq
These are manifestly $\SU(2)$ invariant and extend to homogeneous polynomials on $\R^3\times\R^4$ of bidegree $(0,2)$ and $(2,2)$ respectively. Furthermore, one may readily check that these polynomials are harmonic (separately with respect to $\Xvec$ and
$Q$). Hence, they span $E_8^{triv}$ and $E_{14}^{triv}$ respectively. Noting that $|\Xvec|^2\equiv 1$ on $S^2$, the claim
follows. 
\end{proof}

From now on, we assume that $V_{red}$ lies in the truncated function space $C^\infty_{14}$, so that it has the structure prescribed by Proposition \ref{prop1}. 

Recall that, in the standard Skyrme model, the interaction potential for well separated skyrmion pairs can be modelled using the dipole formalism \cite{mansut}: far from its centre, a unit skyrmion looks like the field induced in the linearization of the Skyrme model about the vacuum, $U=1$, by an orthogonal triplet of scalar dipoles placed at the skyrmion's centre. The interaction potential for a skyrmion pair with relative position $\Xvec$ and orientation $Q$ can then be approximated by the interaction energy of a pair of triplets of dipoles held at relative displacement $\Xvec$ and orientation $Q$, interacting via the linear theory. This approximation introduces another useful constant associated with the unit skyrmion, namely the strength of the (necessarily equal) dipoles. In practice this is determined numerically by reading off a coefficient $C$ in the large $r$ asymptotics of the skyrmion profile function. This formalism is readily adapted to the lightly bound Skyrme model, producing an interaction potential of the form \eqref{Vredansatz} with
\bea
V_0(r) &=& 0 \nonumber \\
V_1(r) &=& -8\pi C^2(1-\alpha)\left(\frac{m}{r^2}+\frac{1}{r^3}\right)e^{-mr} \nonumber \\
\label{dipolepot}
V_2(r) &=& 8\pi C^2(1-\alpha)\left(\frac{m^2}{r}+\frac{3m}{r^2}+\frac{3}{r^3}\right)e^{-mr}.
\eea
The dipole strength (for $\alpha=0.95$ and $m=1$) is found numerically to be $C\approx14.58$. These formulae reproduce the usual prediction of attractive and repulsive channels for well-separated skyrmions. That is, $V_{red}$ is maximally attractive (increases fastest with $|\Xvec|$) if the orientations of the skyrmions differ by a rotation by $\pi$ about any direction orthogonal to $\Xvec$, is maximally repulsive if the orientations differ by a rotation by $\pi$ about $\Xvec$, and is nonmaximally repulsive if their orientations are equal. We refer to these three situations as the attractive, repulsive and product channels respectively.

The existence of these three channels allows us to fix the functions $V_0,V_1,V_2$ numerically by conducting scattering simulations of skyrmion pairs in the full  field theory, in similar fashion to Salmi and Sutcliffe's work on the $(2+1)$ dimensional model \cite{salsut}. We begin with a Skyrme field of the form
\beq\label{ac}
U_a(x_1,x_2,x_3)=U_H(x_1+\frac{s}{2},x_2,x_3)U_H(-(x_1-\frac{s}{2}),-x_2,x_3)
\eeq
where $s>0$ is large and $U_H$ is a unit hedgehog skyrmion defined (numerically) in a ball of radius less than $s/2$ (so 
$U_H(\xvec)=1$ for all $|\xvec|\geq s/2$, and the product above commutes). Such a field represents a pair of skyrmions located at $\xvec=(\pm s/2,0,0)$, that is, with separation $s$, in the attractive channel. Here, and henceforth, we define the skyrmion positions of a Skyrme field
$U:\R^3\ra \SU(2)$ to be those points where $U=-1$. We now allow $U$ to evolve with time according to the dynamics defined by the
lagrangian \eqref{Skyrme lagrangian}, using the fourth order spatial discretization employed by the energy minimization scheme of
\cite{gilharspe}, and a fourth order Runge-Kutta scheme with fixed time step for the time evolution. This numerical scheme conserved total energy $E=T+V$ to extremely high accuracy,
\beq
\max_t \frac{|E(t)-E(0)|}{E(0)}<2.4\times 10^{-5},
\eeq
for all the dynamical processes presented here. As the dipole model predicts, the skyrmions with these initial data slowly move towards one another, attain a minimum separation, then recede again. By recording their separation $s(t)$ and potential energy $V(t)$ at each time step, we recover a numerical approximation to the
attractive channel interaction potential which, according to \eqref{Vredansatz} is related to $V_0,V_1,V_2$ by
\beq
V_a(s)=V_0(s)-V_1(s)-V_2(s).
\eeq
We then repeat the process with intial data
\bea
U_r(x_1,x_2,x_3)&=&U_H(x_1+\frac{s}{2},x_2,x_3)U_H(-(x_1-\frac{s}{2}),x_2,-x_3)\\
U_p(x_1,x_2,x_3)&=&U_H(x_1+\frac{s}{2},x_2,x_3)U_H(x_1-\frac{s}{2},x_2,x_3)
\eea
which are in the repulsive and product channels respectively. To make the skyrmions approach one another and interact, we now Galilean boost them towards one another at low speed (v=0.1). Note that the reflexion symmetries of the initial data trap these fields in their respective channels for all time. From these numerical solutions we obtain numerical approximations to the repulsive and product channel interaction potentials, which are related to $V_0,V_1,V_2$ by
\bea
V_r(s) &=& V_0(s) -V_1(s)+V_2(s), \\
V_p(s) &=& V_0(s) +3V_1(s) + V_2(s).
\eea
It is clear that $V_a,V_r,V_p$ uniquely determine $V_0,V_1,V_2$ and hence, within the ansatz \eqref{Vredansatz}, $V_{int}$. 

\begin{figure}[htb]
\begin{center}
\includegraphics[scale=0.5]{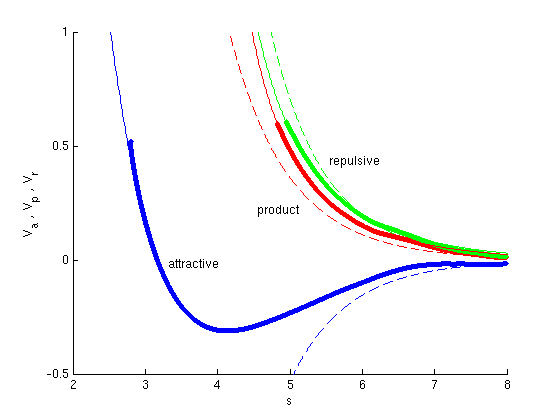}
\end{center}
\caption{Interaction energies of skyrmions pairs with separation $s$ in the attractive (blue), product (red) and repulsive (green) channels. In each case the thick curve represents numerical data extracted from a scattering process, the thin curve is a fit to this, and the dashed curve is the interaction energy predicted by the dipole model.}
\label{fig:potentials}
\end{figure}

Graphs of $V_a,V_r,V_p$, determined numerically as described above, are presented in figure \ref{fig:potentials}. These curves also show the potentials predicted by the dipole model (with dipole strength $C=14.58$). 
Clearly, the dipole formulae \eqref{dipolepot} do \emph{not} provide an accurate quantitative picture of skyrmion interactions in the lightly bound model at any separation where the interactions are not negligible. This is, perhaps, not surprising, since the dipole formalism replaces the full field theory by terms originating only in the quadratic and pion mass potential terms of the lagrangian, and these are precisely the terms which are given very low weighting, $1-\alpha$, in the lightly bound regime. The qualitative predictions of the dipole picture are reliable however: the interaction potentials appear to decay exponentially fast, and the three channels identified have the behaviour predicted (attractive, repulsive, more weakly repulsive). 
For later use, it is convenient to have explicit functions which approximate the numerical data for $V_a,V_r,V_p$. For our purposes, it is important that these functions decay exponentially with $s$ and accurately fit the numerical data for $s\geq s_0$,
where $s_0$ is somewhat smaller then the equilibrium separation defined by $V_a$ (that is, the separation at which $V_a$
is minimal). The behaviour for $s<s_0$ is not so important, provided the formulae introduce a repulsive core interaction, and is, in any case, inaccessible to our numerical scheme (since close approach of lightly bound skyrmions is forbidden in low energy scattering processes). Figure \ref{fig:potentials} also depicts the following fit functions
\begin{equation}
\label{channel functions}
\begin{aligned}
V_a(s) &= \begin{cases} \frac{7.7479-4.5997s+0.8297s^2-0.0473s^3}{1-0.4751s+0.0843s^2+0.0331s^3-0.0049s^4} & 0\leq s<7.096 \\ -94.6178\frac{e^{-s}}{s} & s\geq 7.096, \end{cases} \\
V_r(s) &= \left(\frac{2476}{s}-\frac{20322}{s^2}+\frac{50254}{s^3}\right)e^{-s},\\
V_p(s) &= \left(\frac{2126}{s}-\frac{18325}{s^2}+\frac{47298}{s^3}\right)e^{-s}.\\
\end{aligned}
\end{equation}
Of these, the most elaborate is $V_a$, a Pad\'e approximant on $[0,7.096]$ spliced to an exponentially decaying tail, the splice being chosen so that $V_a$ is continuously differentiable. Unlike $V_r$ and $V_p$, $V_a$ is well defined at $s=0$, where it is chosen to equal the static energy of the axially symmetric $B=2$ solution (a saddle point of the Skyrme energy), obtained numerically by a different scheme, a choice made mainly for aesthetic reasons. 

From now on, we choose $V_{red}$ to be the function defined by \eqref{Vredansatz}, where
\begin{align*}
V_0(s) &= \frac12 V_a(s)+\frac14V_p(s)+\frac14V_r(s) \\
V_1(s) &= \frac14 V_p(s)-\frac14 V_r(s) \\
V_2(s) &= -\frac12 V_a(s) + \frac12 V_r(s)
\end{align*}
and $V_a,V_p,V_r$ are the functions defined in \eqref{channel functions}. It is straightforward to show that this function $V_{red}$ is bounded below as, on physical grounds, it should be.

\subsection{The FCC lattice}

We have seen that the interaction potential $V_{int}$ prefers particles to be in the attractive channel, i.e. such that their relative orientation corresponds to a rotation about an axis perpendicular to their line of separation through angle $\pi$.  It is therefore desirable to find a way to pack them together such that all neighbouring pairs of particles are in the attractive channel.  The face-centred-cubic (FCC) lattice provides a solution to this problem.

The face-centred cubic lattice may be defined to be
\[ \{ (n_1\lambda,n_2\lambda,n_3\lambda)\::\:\mathbf{n}\in\ZZ^3,\,n_1+n_2+n_3 = 0\mod 2 \}, \]
with $\lambda>0$ defining a lattice scale.  The underlying cubic lattice is given by points $(n_1\lambda,n_2\lambda,n_3\lambda)$ for which $n_1,n_2,n_3$ are all even.  Those points for which some of the coordinates $n_i$ are odd lie on faces of the underlying cubic cells.

We assign orientations to these points as follows: those points on the vertices have orientation $1\in\mathrm{SU}(2)$, those on faces perpendicular to the $x$-axis have orientation $\mathbf{i}$, those on faces perpendicular to the $y$-axis have orientation $\mathbf{j}$, and those perpendicular to the $z$-axis have orientation $\mathbf{k}$.  Here we have implicitly identified elements $q\in\mathrm{SU}(2)$ with unit quaternions $q\in\mathbb{H}$, such that $\mathbf{i}=-\ii\sigma_1$, $\mathbf{j}=-\ii\sigma_2$, $\mathbf{k}=-\ii\sigma_3$ and $1$ is the identity matrix.  Put differently, the orientation $q$ of a particle at lattice site $(n_1,n_2,n_3)\lambda$ is such that
\[
R(q) = \begin{pmatrix} (-1)^{n_1}&0&0\\0&(-1)^{n_2}&0\\0&0&(-1)^{n_3}\end{pmatrix}.
\]
The reader may verify that any pair of nearest neighbours, separated by a distance $\lambda\sqrt{2}$, is in the attractive channel.

One might expect that minimizers of the potential energy derived from \eqref{point particle lagrangian} resemble subsets of the FCC lattice.  This was certainly true of all global minima of the Skyrme energy identified in \cite{gilharspe}, and all but one of the local minima.

\subsection{Inertia tensors}

The point particle model \eqref{point particle lagrangian} makes simple predictions for the inertia tensors of lightly bound skyrmions.  These are obtained by calculating the kinetic energy of a rotating and isorotating oriented point cloud.

Let $\{(\mathbf{x}_a,q_a)\}$ be a minimizer of the potential energy derived from \eqref{point particle lagrangian}.  
Choose any pair of angular velocities $(\omega_I,\omega_J)\in\su(2)\oplus\su(2)$. It is useful to identify each $\omega\in \su(2)$ with a vector $\omvec\in\R^3$ by choosing
$-\sfrac\ii2\sigma_j$, $j=1,2,3$, as a basis for $\su(2)$ (so $\omega=-\sfrac\ii2 \omvec\cdot\sigvec$). 
Consider the following configuration, which is isorotating and rotating at constant angular velocity $\omega=(\omega_I,\omega_J)$:
\begin{align*}
(\mathbf{x}_a(t),q_a(t)) &= \exp(\omega t)\cdot (\mathbf{x}_a,q_a) \\
&= \left( R(\exp(\omega_J t))\mathbf{x}_a,\, \exp(\omega_I t)q_a\exp(-\omega_J t) \right).
\end{align*}
We find
\begin{align*}
\dot{\xvec}_a(0) &= \omvec_J\times\xvec_a,\\
\dot{q}_a(0) &= \omega_Iq_a-q_a\omega_J=(\omega_I-q_a\omega_Jq_a^{-1})q_a, 
\end{align*}
whence
\begin{align*}
|\dot{\xvec}_a(0)|^2 &= |\omvec_J|^2|\xvec_a|^2-(\omvec_J\cdot\xvec_a)^2,\\
|\dot{q}_a(0)|^2 &= \frac12\Tr[(\omega_I-q_a\omega_Jq_a^{-1})q_a\, q_a^\dagger(\omega_I^\dagger-q_a\omega_J^\dagger q_a^{-1})]
=|\omvec_I|^2-2\omvec_I\cdot R(q_a)\omvec_J+|\omvec_J|^2. 
\end{align*}
Therefore the kinetic energy is
\beq
\frac12\sum_{a=1}^B \left( M|\dot{\mathbf{x}}_a|^2+L|\dot{q}_a|^2 \right) = \left(\begin{array}{cc} \omvec_I & \omvec_J \end{array}\right)
\Lambda\left(\begin{array}{c} \omvec_I \\ \omvec_J \end{array}\right),
\eeq
where the inertia tensor is
\begin{equation}
\label{point particle inertia tensor}
\Lambda = \sum_{a=1}^B \left( M \left(\begin{array}{c|c} 0_3 & 0_3 \\ \hline 0_3 & |\mathbf{x}_a|^2\mathrm{Id}_3-\mathbf{x}_a\mathbf{x}_a^T \end{array}\right)
+ L \left(\begin{array}{c|c} \mathrm{Id}_3 & -R(q_a) \\ \hline -R(q_a)^T & \mathrm{Id}_3 \end{array}\right) \right).
\end{equation}
The point particle model predicts that this is a good approximation to the inertia tensor of a lightly bound 
degree $B$ skyrmion. We will test this prediction in the next section.

\section{Energy minimizers in the point particle model}
\label{sec:4}\label{sec:numres}
\news

\subsection{Light nuclei}

Having introduced the point particle model for lightly bound skyrmions, in this section we present our results for energy-minimizing configurations of point particles.  We begin by discussing our results for eight particles or fewer, where comparison can be made with energy minima in the lightly bound Skyrme model found in \cite{gilharspe}.

We have developed an iterative zero-temperature annealing algorithm to minimize the energy of a configuration of particles.  We applied this algorithm both to randomly-chosen initial ensembles of particles and to initial ensembles that are subsets of the FCC lattice.  We ran a large number of simulations for each value of $B$, typically obtaining several local energy minima, and record here only the lowest local minimum and up to two closest competitors.  Energies of these local minima with $2\leq B\leq 8$ are presented in table \ref{tab:1 to 8}. The particle ensembles themselves are depicted in figure \ref{fig:ball_stick_3-8}.  The corresponding binding energies in the lightly bound Skyrme model are also recorded in the table.  These are defined to be the energy of the $B$-skyrmion minus $B$ times the energy of the 1-skyrmion.

\begin{table}[htb]
\begin{center}
\begin{tabular}{c|c|c|c}
Name & bonds & particle energy & Skyrme interaction energy \\
\hline
2a & 2 & -0.310 & -0.36\\
3a & 3 & -0.931 & -0.92\\
4a & 6 & -1.862 & -1.71\\
5a & 8 & -2.338 & -2.20\\
5b & 8 & -2.185 & -2.00*\\
6a & 12 & -3.229 & -2.85\\
6b & 11 & -3.117 & -2.87\\
6c & 11 & -3.046 & -2.79*\\
7a & 15 & -4.057 & -3.58*\\
7b & 14 & -3.895 & -3.52\\
8a & 18 & -4.889 & -4.47\\
8b & 18 & -4.869 & -4.37\\
8c & 18 & -4.781 & -4.34*\\
\end{tabular}
\end{center}
\caption{Energies and numbers of bonds of the lowest-energy local minima in the point particle model, and energies of the corresponding lightly bound skyrmions (taken from \cite{gilharspe}, except those marked ${}^*$, which result from new simulations conducted with the same numerical scheme).}
\label{tab:1 to 8}
\end{table}

\begin{figure}[htb]
\begin{center}
\begin{tabular}{cccc}
\includegraphics[scale=0.19]{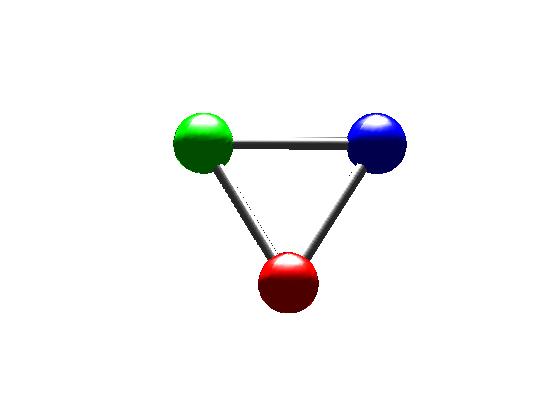}&
\includegraphics[scale=0.19]{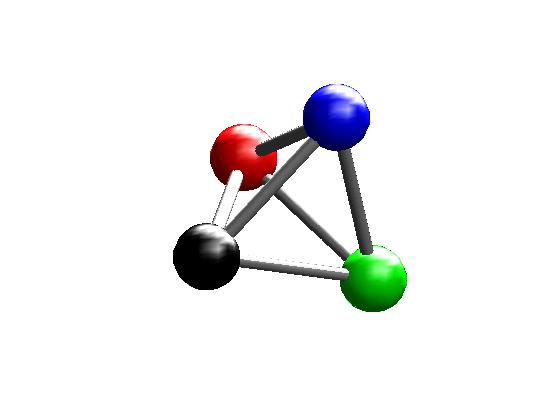}&
\includegraphics[scale=0.19]{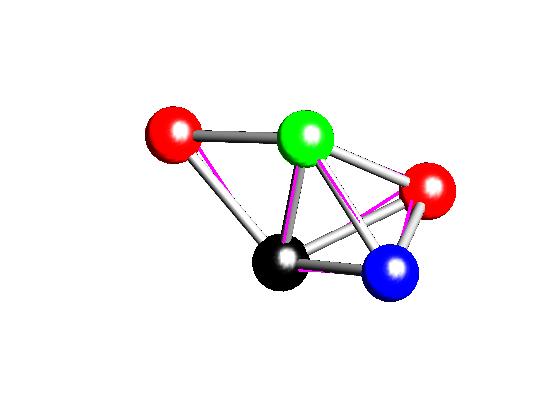}&
\includegraphics[scale=0.19]{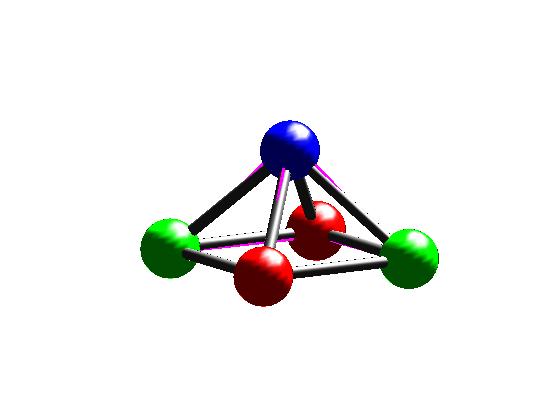}\\
3&4&5(a)&5(b)\\
\includegraphics[scale=0.19]{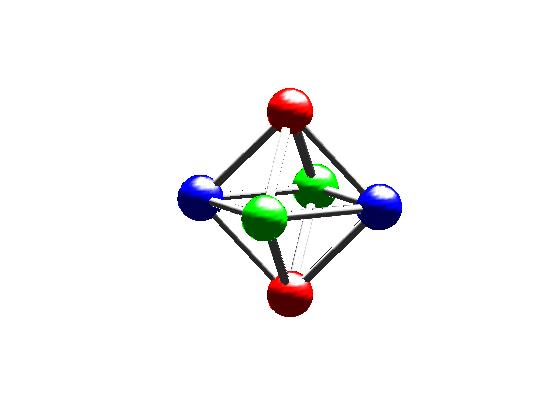}&
\includegraphics[scale=0.19]{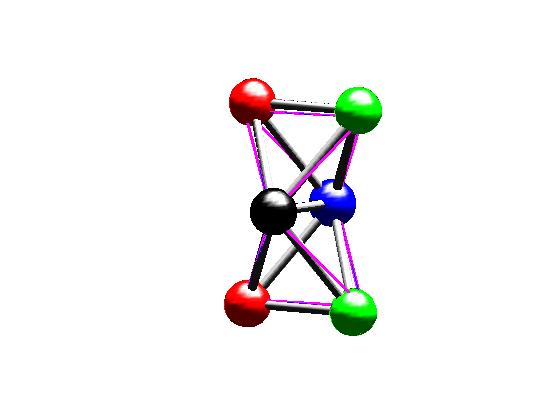}&
\includegraphics[scale=0.19]{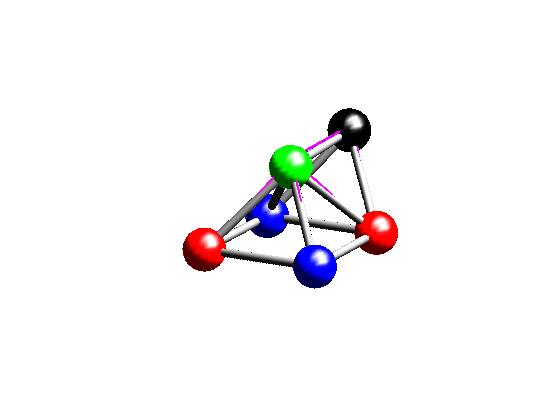}&
\includegraphics[scale=0.19]{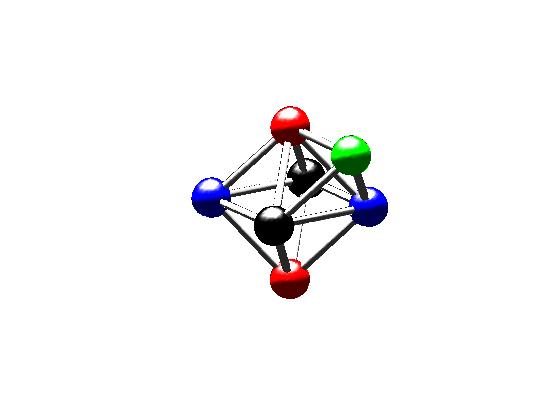}\\
6(a)&6(b)&6(c)&7(a)\\
\includegraphics[scale=0.19]{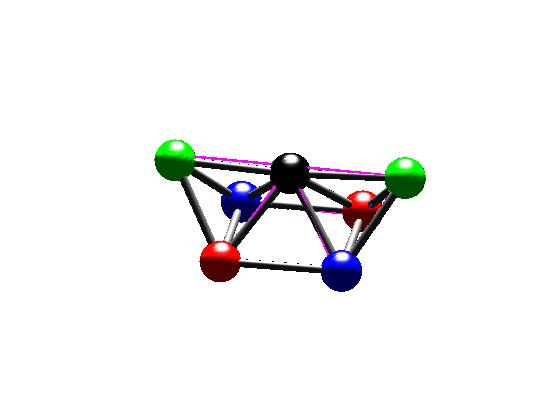}&
\includegraphics[scale=0.19]{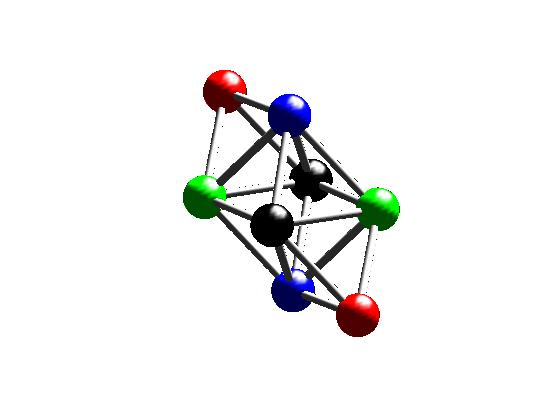}&
\includegraphics[scale=0.19]{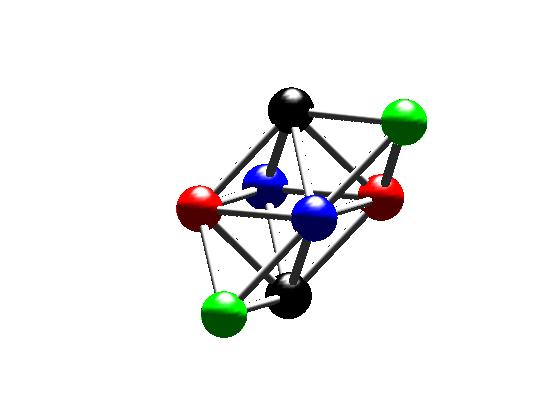}&
\includegraphics[scale=0.2]{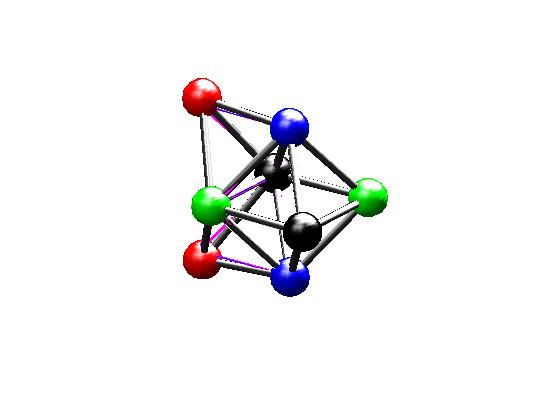}\\
7(b)&8(a)&8(b)&8(c)\\
\end{tabular}
\end{center}
\caption{Local energies minimizers in the point particle model. Each ball is centred on a point skyrmion position $\xvec_a$ and its colour represents the
internal orientation $q_a$. Each picture also depicts the FCC lattice configuration of size $B$ to which the minimizer best fits.
Thick grey line segments indicate interskyrmion bonds no more than 10\% longer than the FCC bond length, while thin magenta line segments show nearest neighbour bonds in the
best fit lattice configuration. In most cases the fit is so good that the thin bonds are not visible. They show quite clearly on 5(a), 5(b), 6(b), 6(c), 7(b) and 8(c) however.}
\label{fig:ball_stick_3-8}
\end{figure}

Our results are almost entirely consistent with the results obtained for the lightly bound Skyrme model in \cite{gilharspe}.  For 
$1\leq B\leq 5$  we obtained the same global minima as in the lightly bound Skyrme model.  For $B=6,7,8$ multiple local minima were previously obtained in the lightly bound Skyrme model.  All of these occured as local minima in the point particle model.  For $B=7,8$ the ordering of energies in the point particle also agreed with the ordering of energies in the lightly bound Skyrme model.  The only failure of the point particle model is for 6 particles: here the energies of the two lowest-energy local minima appear in the wrong order.

In addition to reproducing previously-known minimizers from the lightly bound Skyrme model our point particle model also predicted some new local minima.  Most interestingly, the global energy-minimizer in the point particle model for $B=7$, labelled $7a$ in
figure \ref{fig:ball_stick_3-8}, did not correspond to any solution of the lightly bound Skyrme model found in \cite{gilharspe}.  Based on this discovery, we constructed an approximate Skyrme field with a similar shape to the point particle energy-minimizer, and minimized its energy using the same numerical scheme that was used in \cite{gilharspe}.  After relaxation this Skyrme field had a lower energy than any of the configurations discovered in \cite{gilharspe}, as predicted by the point particle model.  Thus we have a new candidate global energy minimizer at charge seven.
Similarly, new simulations find local energy minimizers in the lightly bound model of similar shape to $5b$, $6c$ and $8c$, and these have energies ordered exactly as the point particle model predicts (so $E_{8c}>E_{8b}>E_{8a}$, for example). 

In every case, the minimizers found look, to the naked eye, like subsets of the FCC lattice.\footnote{All minimizers in this paper can be found at
{\tt http://www1.maths.leeds.ac.uk/$\sim$pmtdgh/lightlybound}}
It is an interesting problem to measure this property quantitatively. Given an oriented point cloud
$(X,Q)=(\xvec_1,\ldots,\xvec_B,q_1,\ldots,q_B)$, we wish to identify the FCC subset of size $B$ which best approximates it. To do this, 
we consider the orbit of $(X,Q)$ under the
group $S$ of \emph{similitudes} of $\R^3$,
\beq
\R^3\times(0,\infty)\times \SU(2)\ni (\cvec,\lambda,h):\xvec\mapsto R(h)\frac{(\xvec-\cvec)}{\lambda}.
\eeq
For each $s\in S$, we define $d^2$ to be the squared distance from $s\cdot X$ to the FCC lattice, i.e.
\beq\label{dsq}
d(s)^2=\sum_{a=1}^B \min\{|\xvec_a-\nvec|^2\: :\: \nvec\in\Z^3,\: n_1+n_2+n_3=0\, \mod 2\}.
\eeq
Now, given a neighbouring triple of particles in $X$ (a particle $\xvec$, its nearest neighbour $\xvec'$ and next-nearest neighbour
$\xvec''$), we construct a
similitude $s_0$ which maps $\xvec$ to $0$, $\xvec'$ to $(1,1,0)$ and $\xvec''$ to the plane spanned by $(1,1,0)$, $(0,1,1)$. We then solve the gradient flow equation of $d^2:S\ra\R$, with $s(0)=s_0$, to find a local minimum of $d^2$ close to $s_0$. Repeating over all neighbouring triples, we keep the lowest local minimum $s_{min}$ of $d^2$ found
 (note that $d^2$ never has a global minimum since $s\cdot X$ can be made arbitrarily close to $(0,0,0)$ by taking $\lambda$ sufficiently large). In this way we identify the closest FCC subset to $X$ and its root mean square distance from $X$, namely
$d_{RMS}=\sqrt{d(s_{min})^2/B}$. Having found $s_{min}\cdot X$, the FCC colouring rule predicts the internal orientations
$(q_1',\ldots,q_B')$ the particles should have. These should be compared with 
$(q_1h_{min}^{-1},\ldots,q_B h_{min}^{-1})$, bearing in mind that
orientations are defined only up to sign, and that the system is isospin invariant. Thus we minimize
\beq\label{dsqiso}
d_{iso}^2:\SU(2)\ra\R,\qquad g\mapsto \sum_{a=1}^B \min\{|gq_ah_{min}^{-1}- q_a'|^2 , |gq_ah_{min}^{-1}+ q_a'|^2\}
\eeq
over $g\in \SU(2)_I$, again by gradient flow. This gives us a measure of the root mean squared distance of the internal orientations of
the configuration $(X,Q)$ from those imposed by the colouring rule applied to its closest FCC approximant, namely $d_{RMS}^{iso}=\sqrt{d_{iso}^2(g_{min})/B}$. It also allows us to ``coarse grain'' the internal orientations, that is, map each $q_a$ to the element of $\{\pm1,\pm\ivec,\pm\jvec,\pm\kvec\}$ to which $gq_ah_{min}^{-1}$ is closest. We used this method to determine the particle colours and FCC bonds in
figure \ref{fig:ball_stick_3-8}. We will present graphs of $d_{RMS}$ and $d_{RMS}^{iso}$ in the next section.

In addition to comparing energies we have also compared inertia tensors in the point particle and lightly bound Skyrme models.  Under isorotations and rotations $(g,h)\in\mathrm{SU}(2)_I\times\mathrm{SU}(2)_J$ inertia tensors transform as
\[
\Lambda \mapsto \left(\begin{array}{c|c} R(g) & 0 \\ \hline 0 & R(h) \end{array}\right) \Lambda \left(\begin{array}{c|c} R(g)^{-1} & 0 \\ \hline 0 & R(h)^{-1} \end{array}\right).
\]
In comparing the inertia tensors of a charge $B$ skyrmion, obtained by solving the field theory, and a charge $B$ point particle energy minimizer, we must account for the fact that the orientations of these two objects are completely unrelated. We do this
by introducing a standard form for inertia tensors which fixes these symmetries.  We say that an inertia tensor $\Lambda$ is in \emph{standard form} if
\begin{equation}
\Lambda = \left(\begin{array}{ccc|ccc}
\ast & \ast & \ast & \mu_1 & \nu_3 & \nu_2 \\
\ast & \ast & \ast & 0 & \mu_2 & \nu_1 \\
\ast & \ast & \ast & 0 & 0 & \mu_3 \\
\hline
\mu_1 & 0 & 0 & \lambda_1 & 0 & 0 \\
\nu_3 & \mu_2 & 0 & 0 & \lambda_2 & 0 \\
\nu_2 & \nu_1 & \mu_3 & 0 & 0 & \lambda_3
\end{array}\right),
\end{equation}
where
\begin{itemize}
\item $\lambda_1,\lambda_2,\lambda_3$ satisfy $|\lambda_1-\lambda_2|\leq|\lambda_2-\lambda_3|\leq|\lambda_1-\lambda_3|$;
\item if $\lambda_1\neq\lambda_2$ then $\mu_1,\mu_2,\mu_3$ are either all non-negative or all non-positive and $\nu_1,\nu_2,\nu_3$ are either all non-negative or all non-positive; and
\item if $\lambda_1=\lambda_2$ then $\nu_3=0$, $|\mu_1|>|\mu_2|$, and $\mu_1,\mu_2,\mu_3,\nu_1,\nu_2$ are either all non-negative or all non-positive.
\end{itemize}

Any inertia tensor has a matrix of standard form in its $\SU(2)_I\times \SU(2)_J$ orbit, and, in generic cases 
this matrix is unique. 
Note that we have chosen not to define standard form as being a form in which both the upper-left and lower-right blocks of $\Lambda$ are diagonal, even though such a form is arguably simpler than the one described above.  The reason is that the upper-left block of any inertia tensor obtained in the point particle model is proportional to the identity, so diagonalising the upper left block does not fix the isorotation symmetry.
We shall measure the distance between inertia tensors by the distance between their standard forms, using the usual Euclidean norm on the space of real matrices, that is
\beq\label{normdef}
\|\Lambda\|^2:=\Tr(\Lambda^T\Lambda).
\eeq

In table \ref{tab:inertia} the distances between inertia tensors obtained in the point particle and lightly bound Skyrme models are recorded.  
The errors recorded in the table are normalised by dividing through by $\|\Lambda\|$, where $\Lambda$ is the lightly bound Skyrme model inertia tensor.  The configurations chosen in this comparison correspond to global energy minima in the lightly bound Skyrme model.  The values of $L$ and $M$ have been chosen to optimise the agreement between the two models, in other words, to minimize the sum over all chosen configurations of the distance between the lightly bound Skyrme and point particle inertia tensors.  The precise values are
\[ M = 93.09,\quad L=54.30. \]
These are quite close to the values $M_H\approx87.49$ and $L_H\approx 53.49$ obtained directly from the 1-skyrmion \eqref{hedgehog ansatz}.  As with energies, agreement of inertia tensors is generally good (within 6\%), with one exception at baryon number 6.
\begin{table}[htb]
\begin{center}
\begin{tabular}{c|cccccccc}
name & 1 & 2 & 3 & 4 & 5 & 6b & 7a & 8a \\
\hline
error & 1.96\% & 5.96\% & 1.61\% & 1.15\% & 4.65\% & 9.91\% & 1.97\% & 3.20\%
\end{tabular}
\end{center}
\caption{Percentage error in inertia tensors calculated in the point particle model, as compared with the lightly bound Skyrme model}
\label{tab:inertia}
\end{table}

\subsection{Heavier nuclei}

When searching for local energy minima with large numbers of particles, one faces the problem that the number of connected subsets of the FCC lattice grows rapidly with the number of particles, and hence the number of candidate local minima of the energy grows rapidly.  We addressed this problem by seeking only local minima corresponding to FCC lattice subsets with a large number of bonds.  More precisely, we used as initial conditions in our relaxation algorithm only lattice subsets whose number of bonds is at most two less than the maximum possible for the given number of particles.  In the end we found that global energy minima always had at most one less than the maximum number of bonds, so this restriction seems reasonable.

Even with this simplification, the number of initial conditions to consider is large and it is difficult to be sure that enough simulations have been run to find the global energy minimizer.  To solve this problem we separated our minimization algorithm into two stages: in the first stage, a list of distinct lattice subsets is generated, and in the second stage these subsets are relaxed as before.  Our method for telling whether two lattice subsets are distinct is to compute their energy: if two lattice subsets have the same energy to high precision we assume that they are identical and discard one.  In doing so we run the risk that a lattice subset whose energy happens to coincide with another is wrongfully discarded.  For example, the initial FCC subsets used to generate solutions 6b and 6c have exactly the same spectrum of bonds, and hence exactly the same energy. Only after relaxation away from the FCC lattice do their energies separate. To mitigate against this danger we ran extensive simulations up to 16 particles starting from randomly chosen lattice subsets satisfying the bond number constraint; in all cases we obtained the same minimum energy as when we started with a list of distinct lattice subsets.

One distinct advantage of our method is that it makes it easy to identify not just the global energy-minimizer but also local energy minima.  Another is that it allows one to tell with reasonable confidence when sufficiently many lattice subsets have been sampled.  Throughout the procedure the number of occurences of each subset is recorded, and when all of these numbers are above a fixed minimum one may assume that all distinct lattice subsets have been found and terminate the algorithm.

In order to generate lattice subsets to use as initial conditions we developed a crystal-growing algorithm.  Again, this algorithm proceeds in two stages.  In the first stage a connected subset of the FCC lattice is generated iteratively.  This scheme starts with a lattice subset consisting of a single point.  At each step of the iteration a member of the lattice subset is chosen at random and one of its twelve nearest neighbours is chosen at random.  If the neighbour is not already a member of the lattice subset, it is appended, otherwise it is discarded.  This continues until the lattice subset has the required number of particles.  In the second stage of the algorithm the subset is modified so as to increase the number of bonds while maintaining a fixed number of particles.  At each step the algorithm chooses at random one member of the subset and a neighbour of another member.  If the neighbour is not already a member of the subset, and replacing the original member with this neighbour increases the number of bonds, the algorithm makes this replacement; otherwise, nothing happens.  This continues for a fixed number of steps.  At each step of the second stage the lattice subset is recorded, so running the algorithm once generates a large number of lattice subsets.

The crystal-growing algorithm was run repeatedly and distinct subsets with sufficiently many bonds saved until it was deemed that enough lattice subsets had been sampled, according to the above-defined criteria.  The maximum number of bonds and the number of crystals identified satisfying our criteria are recorded in table \ref{tab:crystals}.
\begin{table}[htb]
\begin{center}
\begin{tabular}{c|c|c}
particle & maximum number & number of lattice \\
number & of bonds found & subsets identified \\
\hline
9 & 21 & 46 \\
10 & 25 & 34 \\
11 & 28 & 102 \\
12 & 32 & 84 \\
13 & 36 & 69 \\
14 & 40 & 56 \\
15 & 44 & 53 \\
16 & 48 & 51 \\
17 & 52 & 55 \\
18 & 56 & 66 \\
19 & 60 & 88 \\
20 & 64 & 125 \\
21 & 68 & 151 \\
22 & 72 & 221 \\
23 & 76 & 342 \\
\end{tabular}
\end{center}
\caption{Maximum number of bonds in an FCC lattice subset of given size, and the number of lattice subsets identified by our algorithm with at most two fewer bonds than the maximum.}
\label{tab:crystals}
\end{table}

The output of our algorithm is recorded in table \ref{tab:energies}.  The total number of local energy minima found was huge; in this table we list all local minima whose energy is within 0.1 of the lowest energy found, together with configurations 5b, 6b and 6c.  For the most part, energy minimizers have the maximum number of bonds possible (exceptions in the table are marked by asterisks), and have the most even distribution of particle ``colours" (after coarse graining) possible. A notable exception to both these rules is 23a, which has one less bond than maximal, and a rather uneven colour distribution (8,5,5,5) but is, nonetheless, the lowest energy $B=23$ configuration found. This minimizer also has unusually high symmetry, as can be seen from 
figure \ref{fig:aren't_I_special}, which also depicts the highly symmetric minimizers 10b and 19b.
One should note, however, that the point particle model does not always favour highly symmetric configurations. The $B=13$ configuration,
let us call it 13sym, obtained by augmenting a single point by all its nearest neighbours, for example, has the maximal number of bonds, but has energy $-8.556$, which is much higher than the 13a. It also has a very uneven colour distribution: 4,4,4,1. So for $B=13$, unlike $B=23$, the model prefers to sacrifice symmetry in favour of uniform colour distribution. These two charge 13 configurations are also depicted in figure \ref{fig:aren't_I_special}. Note that all particles in 13a are contained in just two planes of the FCC lattice, a feature it has in common with all global minimizers for $4\leq B\leq 15$.

\begin{table}
\begin{center}
\begin{tiny}
\begin{tabular}{c|c|c|c|c|c|c|c||c|}
     &       & Colour & Classical & Symmetry &     &     & Quantum & {Experiment} \\
Name & Bonds & count  & energy    & group    & $I$ & $J$ & energy  &  \\
\hline
2a & 1 & 1,1,0,0 & -0.310 & $D_2$ & 0 & 1 & 3.813 & ${}^2\mathrm{H}_1$  \\
3a & 3 & 1,1,1,0 & -0.931 & $C_3$ & 1/2 & 1/2 & 1.106 & ${}^{3}\mathrm{He}_{2}$ \\
4a & 6 & 1,1,1,1 & -1.862 & $T$ & 0 & 0 & -1.862 & ${}^{4}\mathrm{He}_{2}$ \\
5a & 8 & 2,1,1,1 & -2.338 & 1 & 1/2 & 1/2 & -1.167 &  \\
5b & 8 & 2,2,1,0 & -2.185 & $C_4$ & 1/2 & 3/2 & -0.700 & ${}^{5}\mathrm{He}_{2}$  \\
6a & 12 & 2,2,2,0 & -3.229 & $O$ & 2 & 1 & 4.275 & \\
6b & 11* & 2,2,1,1 & -3.117 & $D_2$ & 0 & 1 & -2.973 & ${}^{6}\mathrm{Li}_{3}$ \\
6c & 11* & 2,2,1,1 & -3.046 & 1 & 0 & 0 & -3.046 & \\
7a & 15 & 2,2,2,1 & -4.057 & $C_3$ & 1/2 & 1/2 & -3.210 &  \\
8a & 18 & 2,2,2,2 & -4.889 & $D_3$ & 0 & 0 & -4.889 & ${}^{8}\mathrm{Be}_{4}$  \\
8b & 18 & 2,2,2,2 & -4.869 & $C_2$ & 0 & 1 & -4.769 &  \\
9a & 21 & 3,2,2,2 & -5.664 & $C_3$ & 1/2 & 1/2 & -5.024 &  \\
9b & 21 & 3,2,2,2 & -5.598 & 1 & 1/2 & 1/2 & -4.956 &  \\
10a & 25 & 3,3,2,2 & -6.443 & $D_2$ & 0 & 1 & -6.352 &  \\
10b & 24* & 4,2,2,2 & -6.442 & $T$ & 0 & 0 & -6.442 &  \\
11a & 28 & 3,3,3,2 & -7.261 & 1 & 1/2 & 1/2 & -6.736 &  \\
12a & 31* & 3,3,3,3 & -8.081 & $C_2$ & 0 & 0 & -8.081 & ${}^{12}\mathrm{C}_{6}$  \\
12b & 32 & 3,3,3,3 & -8.066 & 1 & 0 & 0 & -8.066 &  \\
13a & 36 & 4,3,3,3 & -9.016 & $C_3$ & 1/2 & 1/2 & -8.575 & ${}^{13}\mathrm{C}_{6}$  \\
14a & 39* & 4,4,3,3 & -9.821 & 1 & 0 & 0 & -9.821 & \\
15a & 43* & 4,4,4,3 & -10.653 & 1 & 1/2 & 1/2 & -10.272 & ${}^{15}\mathrm{N}_{7}$  \\
15b & 42** & 4,4,4,3 & -10.627 & 1 & 1/2 & 1/2 & -10.247 & ${}^{15}\mathrm{N}_{7}$  \\
15c & 43* & 4,4,4,3 & -10.584 & 1 & 1/2 & 1/2 & -10.202 & ${}^{15}\mathrm{N}_{7}$  \\
16a & 48 & 4,4,4,4 & -11.771 & $T$ & 0 & 0 & -11.771 & ${}^{16}\mathrm{O}_{8}$  \\
17a & 51* & 5,4,4,4 & -12.563 & $C_3$ & 1/2 & 1/2 & -12.228 &  \\
18a & 54** & 5,5,4,4 & -13.356 & $C_2$ & 0 & 0 & -13.356 &  \\
18b & 56 & 6,4,4,4 & -13.340 & $C_4$ & 0 & 0 & -13.340 &  \\
19a & 60 & 5,5,5,4 & -14.251 & $C_3$ & 1/2 & 1/2 & -13.951 & ${}^{19}\mathrm{F}_{9}$  \\
19b & 60 & 7,4,4,4 & -14.244 & $O$ & 1/2 & 1/2 & -13.946 & ${}^{19}\mathrm{F}_{9}$  \\
19c & 58** & 5,5,5,4 & -14.178 & 1 & 1/2 & 1/2 & -13.879 & ${}^{19}\mathrm{F}_{9}$ \\
19d & 59* & 5,5,5,4 & -14.164 & 1 & 1/2 & 1/2 & -13.864 & ${}^{19}\mathrm{F}_{9}$  \\
20a & 64 & 5,5,5,5 & -15.194 & 1 & 0 & 0 & -15.194 & ${}^{20}\mathrm{Ne}_{10}$  \\
21a & 68 & 6,5,5,5 & -16.118 & 1 & 1/2 & 1/2 & -15.848 &  \\
22a & 72 & 7,5,5,5 & -17.022 & $C_3$ & 0 & 0 & -17.022 &  \\
23a & 75* & 8,5,5,5 & -17.813 & $C_3$ & 1/2 & 1/2 & -17.568 &  \\
23b & 76 & 6,6,6,5 & -17.778 & $C_2$ & 1/2 & 1/2 & -17.531 &  \\
23c & 75* & 6,6,6,5 & -17.755 & 1 & 1/2 & 1/2 & -17.508 &  \\
23d & 75* & 6,6,6,5 & -17.744 & 1 & 1/2 & 1/2 & -17.498 &  \\
23e & 75* & 6,6,6,5 & -17.724 & 1 & 1/2 & 1/2 & -17.478 &  \\
\end{tabular}

\end{tiny}
\end{center}
\caption{The lowest energy local energy minima in the point particle model.  Asterisks in column 2 indicate that the configuration has one or two fewer bonds than the maximum bond number found for that particle number. Column 3 indicates the number of particles of each internal orientation, after coarse-graining. The classical energy is the potential $V$ of eq.\ \eqref{point particle potential}, and the quantum energy is $V+\Delta$, where $\Delta$ is defined in section \ref{sec:5}.  The experiment column 
identifies the lightest nucleus for given baryon number $B$ if this nucleus has the spin and isospin predicted.}
\label{tab:energies}
\end{table}

\begin{figure}[htb]
\begin{center}
\begin{tabular}{ccc}
\includegraphics[scale=0.19]{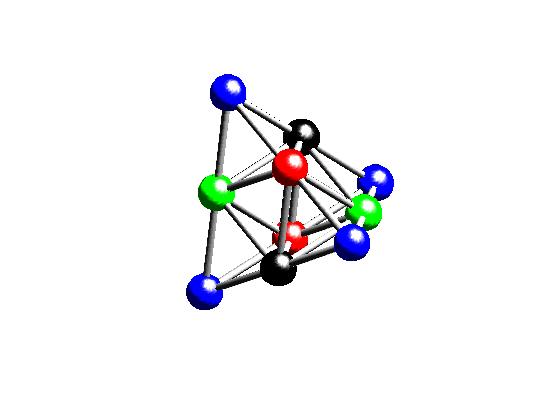}&
\includegraphics[scale=0.19]{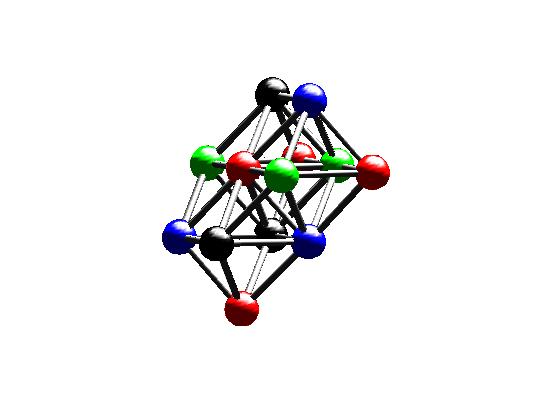}&
\includegraphics[scale=0.19]{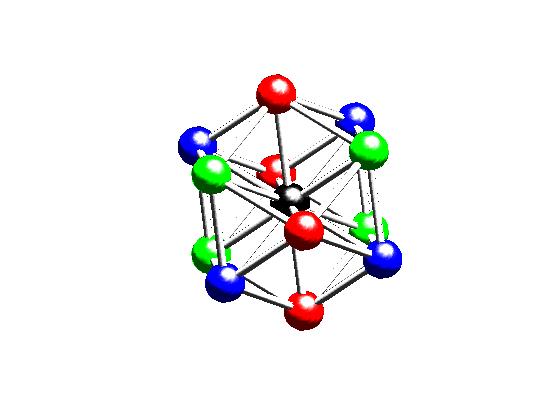}\\
10b&13a&13sym\\
\end{tabular}
\begin{tabular}{cc}
\includegraphics[scale=0.19]{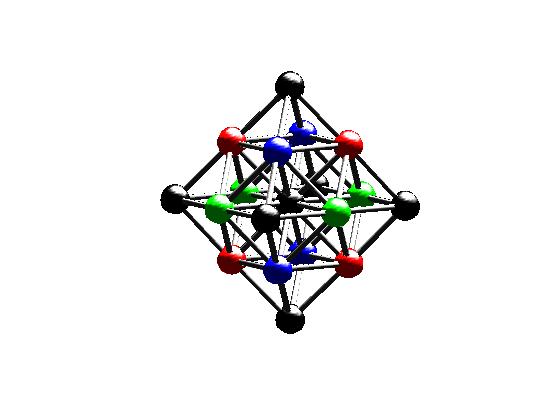}&
\includegraphics[scale=0.19]{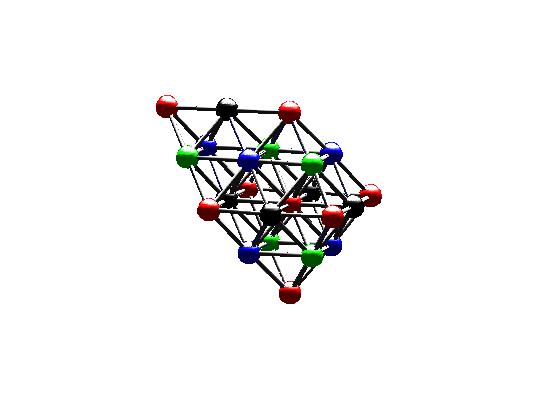}
\\
19b & 23a \\
\end{tabular}
\end{center}
\caption{Selected energy minimizers for $10\leq B\leq 23$. 
23a is a large octohedron with a tetrahedron glued to one face. This unusually symmetric configuration is the global minimizer for $B=23$, despite having less than maximal bond number and rather uneven colour distribution. By contrast, the exceptionally symmetric configuration 13sym has much higher energy than 13a. Also depicted are the local minimizers 10b and 19b, a large tetrahedron and octohedron respectively.}
\label{fig:aren't_I_special}
\end{figure}

The corresponding predictions for nuclear binding energies per nucleon, defined to be $-V_{int}/B$, are plotted in figure \ref{fig:classical BE}.  Here, as in \cite{gilharspe}, energies in table \ref{tab:energies} have been converted to MeV by multiplying with 10.72.  The curve shows that ensembles of 4 and 16 particles have unusually high binding energies, in agreement with nuclear experiment, although these effects are less pronounced in the point particle model than in experiment.  The energy minimizers corresponding to these two peaks are particularly special: they both have tetrahedral symmetry.  Note that our binding energy curve lacks the peak seen at baryon number 12 in the nuclear binding energy curve.

\begin{figure}[htb]
\begin{center}
\includegraphics[scale=1.0]{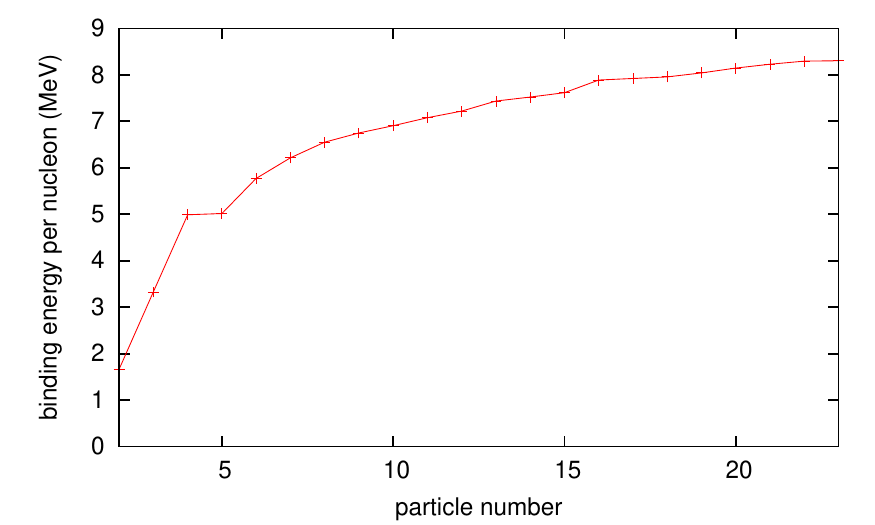}
\end{center}
\caption{Predictions for nuclear binding energies from the point particle model.}
\label{fig:classical BE}
\end{figure}

In order to analyze the overall shape of the energy minimizers we have calculated for each its second moment matrix $M_{ij}$, defined by
\[ \label{secmommat}
M_{ij} := \sum_{a=1}^N (x^i_a-x^i_0)(x^j_a-x^j_0),\quad \mathbf{x}_0:=\frac{1}{N}\sum_{a=1}^N \mathbf{x}_a. \]
This matrix can be decomposed $M=M_1+M_0$, where $M_1=\frac13\Tr(M)\mathrm{Id}_3$ and $M_0$ is traceless.  The trace part $M_1$ provides a measure of the size of the point cloud $\{\xvec_a\}$.  If the cloud is approximately round, then $M$ has nearly equal eigenvalues and the traceless part $M_0$ is close to zero. Therefore $\|M_0\|$ provides a measure of anisotropy of the point cloud (recall that the norm of a matrix is defined in \eqref{normdef}). For a symmetric matrix such as
$M$, 
$\|M\|^2=\lambda_1^2+\lambda_2^2+\lambda_3^2$, where $\lambda_i$ are its eigenvalues. Clearly, $\|M_1\|^2=3\ol{\lambda}^2$ where
$\ol\lambda=(\lambda_1+\lambda_2+\lambda_3)/3$, and the eigenvalues of $M_0$ are $\lambda_i-\ol\lambda$. Hence
\bea
\|M_0\|^2&=&(\lambda_1-\ol\lambda)^2+(\lambda_2-\ol\lambda)^2+(\lambda_3-\ol\lambda)^2\nonumber\\
&=&(\lambda_1+\lambda_2+\lambda_3)^2-2(\lambda_1\lambda_2+\lambda_2\lambda_3+\lambda_3\lambda_1)-3\ol\lambda^2\nonumber\\
&=&6\ol\lambda^2-2(\lambda_1\lambda_2+\lambda_2\lambda_3+\lambda_3\lambda_1)\label{jsg1}
\leq 6\ol\lambda^2=2\|M_1\|^2
\eea
since $M$ is positive definite.
In figure \ref{fig:moments1} we have plotted $\|M_0\|$ against $\|M_1\|$ for the minimizers listed in table \ref{tab:energies}.  Overall there seems to be a downward trend in $\|M_0\|/\|M_1\|$ as $\|M_1\|$ increases, indicating that larger minimizers are closer to being round than small minimizers.  However, even for large nuclei the level of anisotropy is substantial.

The determinant $\det(M_0)$ measures the qualitative nature of the anisotropy.  Let us order the eigenvalues of $M$ so that
$\lambda_1\leq\lambda_2\leq\lambda_3$. If the point cloud is long and thin then $\lambda_1\leq\lambda_2<\ol\lambda<\lambda_3$,
so $M_0$ has two negative eigenvalues and one positive, whence $\det(M_0)>0$. By contrast, if the point cloud is flat and round 
then $\lambda_1<\ol\lambda<\lambda_2\leq\lambda_3$, so $\det(M_0)<0$.  It is useful to define $\mu_i=\lambda_i-\ol\lambda$, the eigenvalues of $M_0$. 
By extremizing the function $\det M_0=\mu_1\mu_2\mu_3$ on the circle obtained by intersecting the sphere of radius
$\|M_0\|$ with the plane $\mu_1+\mu_2+\mu_3=0$, one finds that
\beq\label{jsg2}
 -\frac{1}{3\sqrt{6}}\|M_0\|^3 \leq \det(M_0) \leq \frac{1}{3\sqrt{6}}\|M_0\|^3, 
\eeq
with equality precisely when two of the eigenvalues (of $M_0$ or, equivalently, $M$) coincide. 
Figure \ref{fig:moments1} also displays a plot of $\sqrt[3]{\det(M_0)}$ against $\|M_0\|$ for the minimizers listed in table \ref{tab:energies}.  Interestingly, $\det(M_0)$ is close to either its maximum or its minimum value in the majority of cases, indicating that both extremes of anisotropy are well-represented.

\begin{figure}[htb]
\begin{center}
\begin{tabular}{cc}
(a)&(b)\\
\includegraphics[scale=1.0]{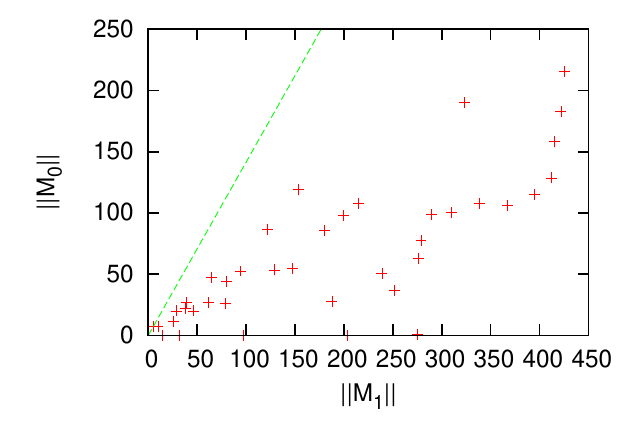}&
\includegraphics[scale=1.0]{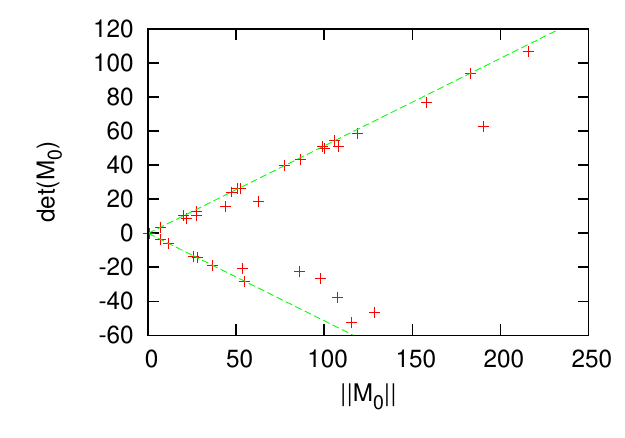}
\end{tabular}
\end{center}
\caption{Graphs showing (a) the anisotropy of energy minimizers listed in table \ref{tab:energies} as a function of their size, and
(b) the type of anisotropy of these energy minimizers.  The dashed lines represent the bounds on $\|M_0\|$ and $\det M_0$ given by
 \eqref{jsg1}, \eqref{jsg2}.}
\label{fig:moments1}
\end{figure}

Just as for $2\leq B\leq 8$, we can measure the distance of each local minimum found from its closest FCC lattice approximant, both in space and in internal (orientation) space, as defined in \eqref{dsq} and \eqref{dsqiso}. The results of this analysis are presented in figure \ref{fig:wwtp?}. With very few exceptions the minimizers match up very closely with FCC subsets, and their internal orientations are very close to the FCC prediction. Note, however, that the optimal lattice scale varies quite significantly with $B$ (rightmost graph), so it is not a good approximation to fix this at the start (to match the FCC bond length to the optimal separation of a single skyrmion pair, for example) and minimize energy only over FCC subsets of that fixed scale. Any attempt to proceed in this way always gets the relative energy ordering of local minima wrong for several values of $B$.

\begin{figure}[htb]
\begin{center}
\includegraphics[scale=0.5]{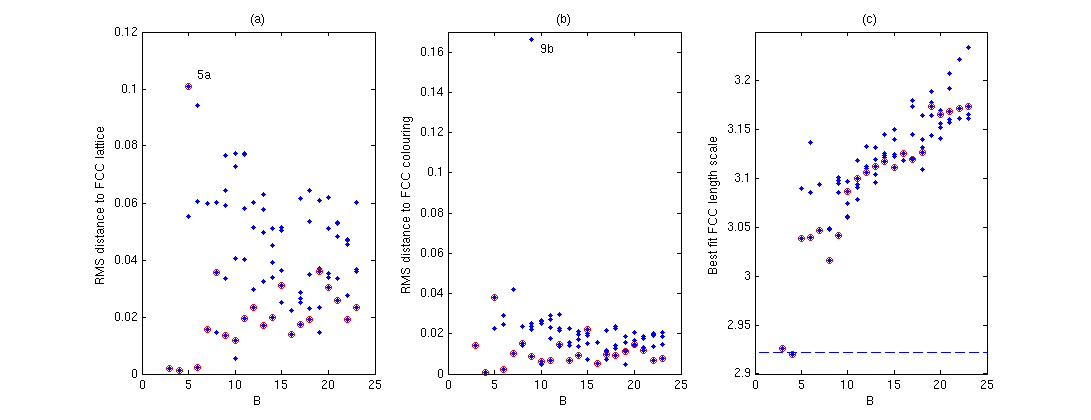}
\end{center}
\caption{Comparison of energy minimizers of the point particle model with subsets of the FCC lattice for baryon number $3\leq B\leq 23$. In each case we shift, scale and rotate the configuration until it matches, as closely as possible, a connected subset of the standard FCC lattice (with integer coordinates). Plot (a) shows the root mean square distance of the particles in each transformed configuration from the FCC lattice, while
plot (b) shows the root mean square distance of their internal orientations from those predicted by the FCC colouring rule. Plot (c) shows the scale factor used. In each case, data corresponding to global energy minima (i.e.\ configurations labelled ``a") 
are circled in red. The dashed line in (c) marks the optimal lattice scale for $B=2$.}
\label{fig:wwtp?}
\end{figure}

\section{Rigid body quantization}
\label{sec:5}\label{sec:rbq}\news

Nuclei are inherently quantum-mechanical, so to make a direct comparison between skyrmions and nuclei it is necessary to include quantum effects in the Skyrme model.  Traditionally this is done semiclassically, treating the classical skyrmions as rigidly rotating and isorotating bodies.  The quantized wavefunction is required to satisfy the Finkelstein-Rubinstein constraints.  In practice these constraints restrict the spin and isospin quantum numbers of a quantized skyrmion.  For example, they guarantee that quantized skyrmions have either half-integer or integer spin and isospin according to whether the baryon number is odd or even.  In cases where solitons have symmetry they yield more nontrivial information.  In the present section we describe how to apply rigid body quantization and the Finkelstein-Rubinstein constraints in the point particle model.  The procedure reduces to a numerical algorithm which we have implemented and applied to the minima presented in the previous section.

\subsection{Finkelstein-Rubinstein constraints}

We begin by recalling the definition of the Finkelstein-Rubinstein constraints (readers not interested in topological details could skip this subsection and continue reading at the start of the next subsection).  The classical configuration space of solitons with baryon number $B$ is the space $\mathcal{S}_B$ of continuous maps $U:\RR^3\to\mathrm{SU}(2)$ of topological degree $B$, satisfying the boundary condition $U(\mathbf{x})\to1$ as $|\mathbf{x}|\to\infty$.  This space is topologically nontrivial: it contains non-contractible loops, so has a nontrivial fundamental group. In fact, $\pi_1(\mathcal{S}_B)=\mathbb{Z}_2$ for
 all $B\in\Z$.  $\mathcal{S}_B$ has a universal covering space $\tilde{\mathcal{S}}_B$ together with a two-to-one map $\pi_{\mathcal{S}}:\tilde{\mathcal{S}}_B\to\mathcal{S}_B$, such that all loops in $\tilde{\mathcal{S}}_B$ are contractible and a loop in $\mathcal{S}_B$ is contractible if and only if it can be lifted to a closed loop in $\tilde{\mathcal{S}}_B$.  The soliton wavefunction is a function $\Psi:\tilde{\mathcal{S}}_B\to\CC$.  The Finkelstein-Rubinstein constraint on $\Psi$ states that for every pair $y$ and $y'$ of distinct points in $\tilde{\mathcal{S}}_B$ such that $\pi_\mathcal{S}(y)=\pi_\mathcal{S}(y')$,
\begin{equation}\label{FR constraint} \Psi(y) = - \Psi(y'). \end{equation}

A configuration of $B$ point particles consists of $B$ vectors $\mathbf{x}_1,\ldots\mathbf{x}_B$ in $\RR^3$ and $B$ elements $q_1,\ldots,q_B\in\mathrm{SU}(2)$.  The energy function disfavours vectors $\mathbf{x}_a$ from being too close, so for practical purposes we may demand that their separations are greater than some fixed minimum $\delta>0$.  Therefore the naive configuration space for the point particle model is
\begin{align}
\tilde{\mathcal{C}}_B &:= \mathrm{SU}(2)^B \times \RR^{3B\ast} \\
\RR^{3B\ast}&:= \big\{(\mathbf{x}_1,\ldots\mathbf{x}_B)\in(\RR^3)^B\::\:|\mathbf{x}^a-\mathbf{x}^b|>\delta \mbox{ whenever } a\neq b \big\}.
\end{align}
The map to Skyrme configuration space is given by a so-called ``relativised product ansatz''
\begin{equation}
\label{relativised product ansatz}
\hat{F}:(\mathbf{x}_1,\ldots,\xvec_B,q_1,\ldots,q_B)\mapsto P\left( \frac{1}{N!}\sum_{\sigma\in\Sigma_B} \prod_{a=1}^B U(\mathbf{x};\mathbf{x}_{\sigma(a)},q_{\sigma(a)}) \right).
\end{equation}
Here $U(\mathbf{x};\mathbf{x}_a,q_a)$ is defined in in \eqref{oriented hedgehog}, $\Sigma_B$ is the group of permutations of the set $\{1,\ldots,B\}$, and
\[ P(U) = \Tr(U^{\dagger}U)^{-1/2} U. \]
Applying $P$ to a sum of products of matrices in $\mathrm{SU}(2)$ yields an $\mathrm{SU}(2)$ matrix, as long as the sum of products is everywhere nonvanishing.  The argument of $P$ in eq.\ \eqref{relativised product ansatz} is nonvanishing if the positions $\mathbf{x}_a$ are sufficiently well-separated, so the right hand side is a map $\RR^3\to\mathrm{SU}(2)$.

The map $\hat{F}:\tilde{\mathcal{C}}_B\to\mathcal{S}_B$ in eq.\ \eqref{relativised product ansatz} is not injective since flipping the sign of any orientation $q_a$, and permuting the particle labels, leave the associated Skyrme field unchanged. To be precise, let $\Sigma_B$ be the group of permutations of $\{1,2,\ldots,B\}$ and $\Z_2=\{1,-1\}$. To each pair $(\sigma,s)\in\Sigma_B\times(\Z_2)^B$, associate the map 
$$
\Phi_{(\sigma,s)}:\tilde{\mathcal{C}}_B\ra
\tilde{\mathcal{C}}_B,\qquad (\xvec_1,\ldots,\xvec_B,q_1,\ldots,q_B)\mapsto (x_{\sigma(1)},\ldots,x_{\sigma(B)},s_{\sigma(1)}q_{\sigma(1)},\ldots,s_{\sigma(B)}q_{\sigma(B)}).
$$
This defines a {right} action of $\Sigma_B \ltimes (\Z_2)^B$ on 
$\tilde{\mathcal{C}}_B$ by homeomorphisms, where the semi-direct product carries group operation
$$
(\sigma,{s})\cdot(\mu,{t})=(\sigma\circ\mu,(s_1t_{\sigma^{-1}(1)},\ldots,s_Bt_{\sigma^{-1}(B)})).
$$
This action leaves $\hat{F}$ invariant, that is, $\hat{F}\circ\Phi_{(\sigma,s)}=\hat{F}$, so  $\hat{F}$ descends to a continuous map $F:\mathcal{C}_B\to\mathcal{S}_B$ where 
\[ \mathcal{C}_B := \tilde{\mathcal{C}}_B/\Sigma_B \ltimes (\ZZ_2)^B \]
is the true point particle configuration space. Since $\tilde{\mathcal{C}}_B$ is simply connected and the action is free, 
 $\tilde{\mathcal{C}}_B$ is the universal cover of $\mathcal{C}_B$ and $\pi_1(\mathcal{C}_B)\cong \Sigma_B\ltimes (\ZZ_2)^B$. Clearly, $\hat{F}=F\circ\pi_\mathcal{C}$ where 
$\pi_\mathcal{C}:\tilde{\mathcal{C}}_B\ra\mathcal{C}_B$ is the canonical projection. 

Choose any pair of points $x_0\in\tilde{\mathcal{C}}_B$, $y_0\in\tilde{\mathcal{S}}_B$ such that $\hat{F}(x_0)=\pi_\mathcal{S}(y_0)$. By a standard theorem of topology (see \cite[pp 61-2]{hat} for
example), $\hat{F}$ has a unique continuous lift $\tilde{F}:\tilde{\mathcal{C}}_B\ra\tilde{\mathcal{S}}_B$ with $\tilde{F}(x_0)=y_0$. The situation is summarized in the following commutative diagram
\beq
\begin{xy}
(0,20)*+{\tilde{\mathcal{C}}_B}="a"; (20,20)*+{\tilde{\mathcal{S}}_B}="b";%
(0,0)*+{\mathcal{C}_B}="c"; (20,0)*+{\mathcal{S}_B}="d";%
{\ar "a";"b"}?*!/_3mm/{\tilde{F}};
{\ar@{->}_{\pi_{\mathcal{C}}} "a";"c"};{\ar@{->}^{\pi_{\mathcal{S}}} "b";"d"};%
{\ar "c";"d"}?*!/_3mm/{{F}};%
{\ar@{->} "a";"d"};?*!/_3mm/{\hat{F}};
\end{xy}.
\eeq
Note that $\wt{F}$ is a lift of $F$.

Any wavefunction $\Psi:\tilde{\mathcal{S}}_B\to\CC$ defines a wavefunction $\psi=\Psi\circ\tilde{F}$ on $\tilde{\mathcal{C}}_B$, which must satisfy some nontrivial constraints derived from the Finkelstein-Rubinstein constraints:
\begin{proposition}
If $x,x'$ are two points in $\tilde{\mathcal{C}}_B$ such that $\pi_\mathcal{C}(x)=\pi_\mathcal{C}(x')$, then
\[ \psi(x') = \mathrm{sgn}(\sigma)\prod_{a=1}^Bs_a\, \psi(x), \]
where $(\sigma,s)\in\Sigma_B\ltimes (\ZZ_2)^B$ is the unique group element that maps $x$ to $x'$.
\end{proposition}
\begin{proof}
This result follows almost directly from two important results of Finkelstein and Rubinstein \cite{finrub}.  First, if $x\in\tilde{\mathcal{C}}_B$, and $t^a\in (\ZZ_2)^B$ is the transformation that changes the sign of $q_a$ (only) and $\alpha$ is a path in $\tilde{\mathcal{C}}_B$ from $x$ to $\Phi_{(\Id,t^a)}(x)$ then $F\circ\pi_C\circ\alpha$ is non-contractible.  Second, if $\sigma\in\Sigma_B$ is a transposition and $\beta$ is a path in $\tilde{\mathcal{C}}_B$ from $x$ to $\Phi_{(\sigma,1)}(x)$ then $F\circ\pi_C\circ\alpha$ is also non-contractible.  Thus the constraint \eqref{FR constraint} implies that
\[ \psi(\Phi_{(\Id,t^a)}(x)) = -\psi(x)\qquad\mbox{and} \qquad\psi(\Phi_{(\sigma,1)}(x))=-\psi(x). \]
Now any element of $\Sigma_B\ltimes(\ZZ_2)^B$ can be written as a product of sign flips and transpositions, so the claim follows.
\end{proof}

\subsection{Rigid body quantization}

In rigid body quantization, motion is restricted to the rotation-isorotation orbit of a fixed minimum $x=(\xvec_1,\ldots,\xvec_B,q_1,\ldots,q_B)$ of the classical energy.  Thus the classical configuration space is taken to be $G=\mathrm{SU}(2)_I\times\mathrm{SU}(2)_J$ with each $(g,h)\in G$ identified with
\[ (g,h)\cdot (\mathbf{x}_a,q_a) = (R(h)\mathbf{x}_a,hq_ag^{-1})\in\tilde{\mathcal{C}}_B. \]
The wavefunction $\psi:G\to\CC$ is required to solve a Schr\"odinger equation $\hat{H}\psi=E\psi$, where $\hat{H}$ is (up to a constant factor) the Laplacian operator
on $G$ associated with the left invariant metric $\Lambda$, the inertia tensor of $x$.

In order to model a nucleus of definite spin and isospin one assumes that $\psi$ is an eigenstate of the total isospin and spin operators with isospin $I$ and spin $J$.  This is consistent with the Schr\"odinger equation because the hamiltonian commutes with these operators.  By the Peter-Weyl theorem, any such $\psi$ is a finite sum of functions of the form
\[ \psi(g,h) = \langle w, \rho_I(g)\otimes\rho_J(h) v\rangle,\quad v,w\in V_{I,J}:=\CC^{2I+1}\otimes\CC^{2J+1}, \]
where for $\ell\in\frac{1}{2}\N$, $\rho_\ell:\mathrm{SU}(2)\to\mathrm{SU}(2\ell+1)$ denotes the spin-$\ell$ representation of $\mathrm{SU}(2)$.  For each $w\in V_{I,J}$ denote by $V^{(w)}$ the
subspace of functions with $w$ fixed. Clearly $V^{(w)}\cong V_{I,J}$ for all $w\neq 0$. Furthermore, $\hat{H}$ preserves $V^{(w)}$, and its action on every $V^{(w\neq0)}$ is unitarily equivalent. Hence we may, without loss of generality, fix $w\neq 0$, and represent $\hat{H}$ by a linear operator $H_{I,J}$ on $V^{(w)}\cong V_{I,J}$. To write this operator down explicitly, it is useful to introduce the usual basis for $\g=\su(2)_I\oplus\su(2)_J$, namely
\beq
K_i=-\frac\ii2\sigma_i\oplus 0,\qquad K_{i+3}=0\oplus\left(-\frac\ii2\sigma_i\right),\qquad i=1,2,3,
\eeq
with respect to which $\Lambda$ is a symmetric $6\times 6$ real matrix. Denote its entries $\Lambda_{ab}$ and those of its inverse $\Lambda^{ab}$. Then $\hat{H}$ acts on $V^{(w)}\cong V_{I,J}$ as
\begin{equation}
\label{algebraic hamiltonian}
H_{I,J}:v\mapsto -\frac{\hbar^2}{2}\Lambda^{ab}\rho^\ast_{I,J}(K_a)\rho^\ast_{I,J}(K_b)v,
\end{equation}
where  $\rho^{\ast}_{I,J}:\mathfrak{su}(2)_I\oplus\mathfrak{su}(2)_J\to \mathfrak{su}((2I+1)(2J+1))$ is the Lie algebra representation associated to $\rho_I\otimes\rho_J$:
\[ \rho^\ast_{I,J}(K_a) = \begin{cases} \rho^\ast_I(K_a)\otimes\mathrm{Id}_{2J+1} & a=1,2,3 \\ \mathrm{Id}_{2I+1}\otimes\rho^\ast_J(K_a) & a=4,5,6. \end{cases} \]

Suppose that $(g_0,h_0)\in\mathrm{SU}(2)_I\times\mathrm{SU}(2)_S$ is a symmetry of $x$, i.e.\ that there exists a $(\sigma,s)\in\Sigma_B\ltimes (\ZZ_2)^B$ such that
\[ (R(h_0)\mathbf{x}_a,h_0q_ag_0^{-1}) = (\mathbf{x}_{\sigma(a)},s_{\sigma(a)}q_{\sigma(a)}). \]
Then the Finkelstein-Rubinstein constraints described in the previous section imply that
\[ \psi(gg_0,hh_0) = \mathrm{sgn}(\sigma)\left(\prod_{a=1}^Bs_a\right)\, \psi(g,h)\quad\forall (g,h)\in\mathrm{SU}(2)_I\times\mathrm{SU}(2)_J. \]
This in turn implies that
\begin{align}  
\label{FR algebraic}
\rho_I(g_0)\otimes\rho_J(h_0) v &= \chi(g_0,h_0)v, \\
\chi(g_0,h_0) &:= \mathrm{sgn}(\sigma)\prod_{a=1}^Bs_a.
\end{align}
Thus each element of the symmetry group of $x$ determines a linear constraint on $v$, and $v$ therefore must belong to the subspace $V^{(x)}_{I,J}\subseteq\CC^{2I+1}\otimes\CC^{2J+1}$ on which all of these constraints are satisfied simultaneously.  Therefore, to find the lowest energy quantized state of a configuration with isospin $I$ and spin $J$ one needs to find the smallest eigenvalue $\Delta$ of the restriction of $H_{I,J}$ to $V^{(x)}_{I,J}$.

One important consequence of equation \eqref{FR algebraic} is that nucleons have half-integer spin and isospin.  This can be deduced using the trivial symmetries $(g_0,h_0)=(1,-1)$ and $(-1,1)$, which are symmetries of any configuration.  Both of these transformations negate all of the orientations $q_a$, so the sign appearing on the right of eq.\ \eqref{FR algebraic} is $(-1)^B$.  Now $\rho_I(-1)$ equals $-\mathrm{Id}_{2I+1}$ if $I$ is half-integer and $\mathrm{Id}_{2I+1}$ if $I$ is integer, so
$V^{(x)}_{I,J}=\{0\}$ if $I$ is half integer and $B$ is even, or if $I$ is integer and $B$ is odd. Hence, there are no half integer isospin energy eigenstates when $B$ is even, and
no integer isospin energy eigenstates when $B$ is odd. Similar comments apply to spin.

\subsection{Two particles}

We now illustrate the quantization procedure for the simple example of two particles.  After rotation and centreing, the energy minimizer 2a is
\begin{equation}
\label{2a}
x_{2a}=(\mathbf{x}_1,\mathbf{x}_2,q_1,q_2) = \left( -\frac{\lambda}{\sqrt{2}}\mathbf{e}_1,\,\frac{\lambda}{\sqrt{2}}\mathbf{e}_1,\,1,\,\mathbf{k} \right),
\end{equation}
where the lattice scale parameter $\lambda$ takes the value $2.9$ to minimize energy and $\mathbf{e}_1=(1,0,0)$.

This configuration has $D_2$ dihedral symmetry, and the nontrivial elements of the symmetry group are $(\mathbf{i},\mathbf{j})$, $(\mathbf{i},\mathbf{i})$,  and $(1,\mathbf{k})$.  The actions of these transformations, and the corresponding signs $\chi(g,h)$, are as follows:
\begin{align*}
(\mathbf{i},\mathbf{j}):(\mathbf{x}_1,\mathbf{x}_2,q_1,q_2)&\mapsto \left( \frac{\lambda}{\sqrt{2}}\mathbf{e}_1,\,-\frac{\lambda}{\sqrt{2}}\mathbf{e}_1,\,\mathbf{k},\,1\right) & \chi(\mathbf{i},\mathbf{j}) &= -1 \\
(\mathbf{i},\mathbf{i}):(\mathbf{x}_1,\mathbf{x}_2,q_1,q_2)&\mapsto \left( -\frac{\lambda}{\sqrt{2}}\mathbf{e}_1,\,\frac{\lambda}{\sqrt{2}}\mathbf{e}_1,1,\,\,-\mathbf{k} \right) & \chi(\mathbf{i},\mathbf{i}) &= -1  \\
(1,\mathbf{k}):(\mathbf{x}_1,\mathbf{x}_2,q_1,q_2)&\mapsto \left( \frac{\lambda}{\sqrt{2}}\mathbf{e}_1,\,-\frac{\lambda}{\sqrt{2}}\mathbf{e}_1,\,\mathbf{k},\,-1 \right)  & \chi(1,\mathbf{k}) &= 1. \\
\end{align*}
We will briefly explain how these signs $\chi(g_0,h_0)$ have been determined.  The first transformation $(\mathbf{i},\mathbf{j})$ permutes the two particles but does not change any signs.  It therefore has $\mathrm{sgn}(\sigma)=-1$, $s_1=s_2=1$, and consequently $\chi(\mathbf{i},\mathbf{j})=-1$.  The second does not permute the particles but does change the sign of exactly one orientation, so has $\chi(\mathbf{i},\mathbf{i}) = -1$ .  The third permutes the particles and changes one sign, so has  $\chi(1,\mathbf{k})=(-1)^2=1$.

Now we determine the subspaces $V^{(2a)}_{I,J}$ allowed by the constraint \eqref{FR algebraic}.  As explained above, the isospin and spin quantum numbers $I,J$ are necessarily integers.  Since some of the symmetries have negative signs in \eqref{FR algebraic} states with $(I,J)=(0,0)$ are forbidden.  Thus we consider the possibilities $(I,J)=(1,0)$ or $(0,1)$.
If $V^{(2a)}_{1,0}$ or $V^{(2a)}_{0,1}$ is nontrivial, states with higher spin and isospin certainly have higher energy.  We need to use the spin 1 representation $\rho_1$ of $\mathrm{SU}(2)$, for which \[ \rho_1(\mathbf{i}) = \begin{pmatrix}1&0&0\\0&-1&0\\0&0&-1\end{pmatrix},\,\rho_1(\mathbf{j}) = \begin{pmatrix}-1&0&0\\0&1&0\\0&0&-1\end{pmatrix},\,\rho_1(\mathbf{k}) = \begin{pmatrix}-1&0&0\\0&-1&0\\0&0&1\end{pmatrix}. \]
In the case $(I,J)=(1,0)$ the constraints \eqref{FR algebraic} reduce to $\rho_1(\mathbf{i})v=-v$, so their solution space is
\[ V^{(2a)}_{1,0} = \{ (0,v_2,v_3)^T\::\: v_2,v_3\in\CC\}. \]
In the case $(I,J)=(0,1)$ the constraints \eqref{FR algebraic} say that $\rho_1(\mathbf{k})v=v$ and $\rho_1(\mathbf{i})v=\rho_1(\mathbf{j})v=-v$, and their solution space is
\[ V^{(2a)}_{0,1} = \{ (0,0,v_3)^T\::\: v_3\in\CC\}. \]

The inertia tensor of the configuration \eqref{2a} is
\[
\Lambda = \left( \begin{array}{ccc|ccc} 
2L & 0 & 0 & 0 & 0 & 0 \\
0 & 2L & 0 & 0 & 0 & 0 \\
0 & 0 & 2L & 0 & 0 & -2L \\
\hline
0 & 0 & 0 & 2L & 0 & 0 \\
0 & 0 & 0 & 0 & 2L+M\lambda^2 & 0 \\
0 & 0 & -2L & 0 & 0 & 2L+M\lambda^2
\end{array}\right),
\]
{whose inverse is
\[
\Lambda^{-1} = \left( \begin{array}{ccc|ccc} 
\frac{1}{2L} & 0 & 0 & 0 & 0 & 0 \\
0 & \frac{1}{2L} & 0 & 0 & 0 & 0 \\
0 & 0 & \frac{2L+M\lambda^2}{2LM\lambda^2} & 0 & 0 & \frac{1}{M\lambda^2} \\
\hline
0 & 0 & 0 & \frac{1}{2L} & 0 & 0 \\
0 & 0 & 0 & 0 & \frac{1}{2L+M\lambda^2} & 0 \\
0 & 0 & \frac{1}{M\lambda^2} & 0 & 0 & \frac{1}{M\lambda^2}
\end{array}\right).
\]}
The representation $\rho_1^\ast$ of the Lie algebra $\mathfrak{su}(2)$ is
\[ 
\rho_1^\ast\left(\sfrac{\mathbf{i}}{2}\right) = \begin{pmatrix} 0&0&0\\0&0&-1\\0&1&0\end{pmatrix},\,
\rho_1^\ast\left(\sfrac{\mathbf{j}}{2}\right) = \begin{pmatrix} 0&0&1\\0&0&0\\-1&0&0\end{pmatrix},\,
\rho_1^\ast\left(\sfrac{\mathbf{k}}{2}\right) = \begin{pmatrix} 0&-1&0\\1&0&0\\0&0&0\end{pmatrix}.
\]
In the case $(I,J)=(1,0)$ the hamiltonian defined in \eqref{algebraic hamiltonian} is
\begin{align*}
H_{1,0} &= -\frac{\hbar^2}{4L}\left(\rho_1^\ast\left(\sfrac{\mathbf{i}}{2}\right)^2+\rho_1^\ast\left(\sfrac{\mathbf{j}}{2}\right)^2\right)-\frac{\hbar^2(2L+M\lambda^2)}{4LM\lambda^2}\rho_1^\ast\left(\sfrac{\mathbf{k}}{2}\right)^2 \\
&=\begin{pmatrix}\frac{\hbar^2}{2L}+\frac{\hbar^2}{2M\lambda^2}&0&0\\0&\frac{\hbar^2}{2L}+\frac{\hbar^2}{2M\lambda^2}&0\\0&0&\frac{\hbar^2}{2L}\end{pmatrix}.
\end{align*}
In the case $(I,J)=(0,1)$ it is
\begin{align*}
H_{0,1} &= -\frac{\hbar^2}{4L}\rho_1^\ast\left(\sfrac{\mathbf{i}}{2}\right) - \frac{\hbar^2}{4L+2M\lambda^2}\rho_1^\ast\left(\sfrac{\mathbf{j}}{2}\right)^2 -\frac{\hbar^2}{2M\lambda^2}\rho_1^\ast\left(\sfrac{\mathbf{k}}{2}\right)^2 \\
&= \begin{pmatrix}\frac{\hbar^2}{4L+2M\lambda^2}+\frac{\hbar^2}{2M\lambda^2}&0&0\\0&\frac{\hbar^2}{4L}+\frac{\hbar^2}{2M\lambda^2}&0\\0&0&\frac{\hbar^2}{4L}+\frac{\hbar^2}{4L+2M\lambda^2}\end{pmatrix}.
\end{align*}
Therefore, the lowest eigenvalue of the restriction of $H$ to $V^{(2a)}_{I,J}$ is $\hbar^2/2L$ in the case $(I,J)=(1,0)$ and $\hbar^2/4L+\hbar^2/(4L+2M\lambda^2)$ in the case $(I,J)=(0,1)$.  Since
\[ \frac{\hbar^2}{4L}+\frac{\hbar^2}{4L+2M\lambda^2}<\frac{\hbar^2}{2L}, \]
the groundstate has isospin 0 and spin 1.  The quantum mechanical correction to the energy is
\[ \Delta = \frac{\hbar^2}{4L}+\frac{\hbar^2}{4L+2M\lambda^2} \approx 4.123, \]
using the value $\lambda=2.922$ of the lattice corresponding the minimum of $V$.  
Hence, rigid body quantization of the $B=2$ energy minimizer correctly reproduces the spin and isospin of the deuteron.

\subsection{Automation of the quantization procedure}
Most of the quantization procedure described above is linear algebraic and so easily automated, but determining the symmetries of a configuration and the associated signs arising in the Finkelstein-Rubinstein constraints can be tricky.  We have developed an algorithm that finds symmetries of configurations of particles, and hence have been able to automate the entire quantization procedure.

Our symmetry-finding algorithm first identifies rotational symmetries of the set of particle positions, and then determines which of these can be lifted to symmetries in $\mathrm{SU}(2)_I\times\mathrm{SU}(2)_S$.  To find spatial symmetries it first translates the configuration so that its centre of mass is at the origin and rotates it so that its second moment matrix
\eqref{secmommat} is diagonal.  It then treats separately three cases corresponding to different degeneracies of the diagonal entries, i.e.\ eigenvalues, of the second moment matrix.

If the eigenvalues are all different then the symmetry group is a subgroup of the group $D_2\in\mathrm{SO}(3)$ consisting of rotations about the three coordinate axes through $\pi$.  Each element of this group is applied to the particle positions and a distance between the resulting configuration and the original configuration is measured (taking account of permutations).  If the distance is sufficiently small then the group element is accepted as a symmetry.

If two eigenvalues are equal and the third distinct then the spatial symmetry group is a subgroup of $O(2)$.  Since, by observation, none of the minimizers with three or more particles have continuous symmetry, the symmetry group is assumed to be dihedral or cyclic.  To identify cyclic symmetries, rotations through angle $\theta$ are applied to the configuration and the distance from the resulting configuration to the original configuration measured as a function of $\theta$.  Minima of this function close to zero are interpreted as symmetries.  Dihedral symmetries are identified in a similar way.

If all three eigenvalues are equal then the symmetry group is likely to be a discrete subgroup of $\mathrm{SO}(3)$ which is neither dihedral nor cyclic.  Since icosahedral symmetry is not compatible with the FCC lattice our algorithm works on the assumption that the symmetry group is either the octahedral group $O$ or the tetrahedral group $T< O$.  It computes the fourth moment tensor
\[ M^{(4)}_{ijkl} = \sum_{a=1}^B x_i^a x_j^a x_k^a x_l^a \]
and finds minima or maxima of the function $S^2\to\RR$ defined by
\[ \mathbf{n}\mapsto M^{(4)}_{ijkl}n_in_jn_kn_l,\quad \mathbf{n}\cdot\mathbf{n}=1. \]
This is a polynomial function on the sphere of degree less than or equal to 4 containing terms of even degree only.  It is known that there are only two linearly independent functions of this type with tetrahedral symmetry, namely the constant function and the function
\[ \mathbf{n}\mapsto n_1^4+n_2^4+n_3^4. \]
The latter furthermore has octahedral symmetry and its maxima are at the points where the coordinate axes intersect the sphere.  Therefore, if the configuration has either octahedral or tetrahedral symmetry and the function constructed from the fourth moment tensor is non-constant then either its maxima are at mutually-orthogonal points on the sphere or its minima are.  The algorithm seeks either a pair of orthogonal maxima or a pair of orthogonal minima and rotates these to lie on two of the coordinate axes.  It then tests whether each element of the octahedral group is a symmetry of the configuration.

Once the configuration's rotational symmetries are known, the algorithm determines whether each lifts to a full rotation-isorotation symmetry of $(\xvec_1,\ldots,\xvec_B,q_1,\ldots,q_B)$. Given a rotational symmetry $R$, we choose $h\in \SU(2)$ such that $R=R(h)$. Being a rotational symmetry means precisely that $R(h)\xvec_a=\xvec_{\sigma(a)}$ for some permutation $\sigma$. Note that spatial rotations change the orientations also, $q_a\mapsto q_a h^{-1}$. The spatial symmetry $h$ lifts to a full symmetry if there exists $g\in\SU(2)$ such that
$gq_ah^{-1}=\pm q_{\sigma(a)}$ for all $a$. If such an isorotation $g$ exists, it is unique up to sign. In fact it must be
$g=q_{\sigma(1)}hq_1^{-1}$ (or minus this). Hence, our algorithm computes $q_a':=q_{\sigma(1)}hq_1^{-1}q_ah^{-1}$ for each $a=2,\ldots,B$
and tests whether $q_a'=\pm q_{\sigma(a)}$ (to some numerical tolerance) for all $a$. If so, $(g,h)$ is accepted as a full symmetry, and its FR factor $\chi(g,h)$ is readily computed from the sign of the permutation $\sigma$ and the signs occuring in $q_a'=\pm q_a$.
If not, the spatial symmetry $R(h)$ is discarded.

The output of our algorithm is recorded in table \ref{tab:energies}.  For each classical energy minimizer we record the spin and isospin quantum numbers corresponding to the ground state, and the quantum mechanical energy (defined to be the sum of the classical energy $V$ and the $O(\hbar^2)$ correction $\Delta$ calculated by our algorithm).  The numerical value of $\hbar$ is fixed by the calibration proposed in \cite{gilharspe}.  We have been using energy and length units $F_\pi/4g\sqrt{1-\alpha}$ and $2\sqrt{1-\alpha}/F_\pi g$.  In natural units Planck's constant is 1, so in our units it equals
\[ \frac{4g\sqrt{1-\alpha}}{F_\pi}\frac{F_\pi g}{2\sqrt{1-\alpha}}=2g^2. \]
In \cite{gilharspe} the dimensionless parameter $g$ was determined to be 3.96 by comparing the charge radii of the one-skyrmion and the proton, so the numerical value for $\hbar^2$ is $4\times3.96^4\approx799.5$.

\begin{figure}[htb]
\begin{center}
\includegraphics[scale=1.0]{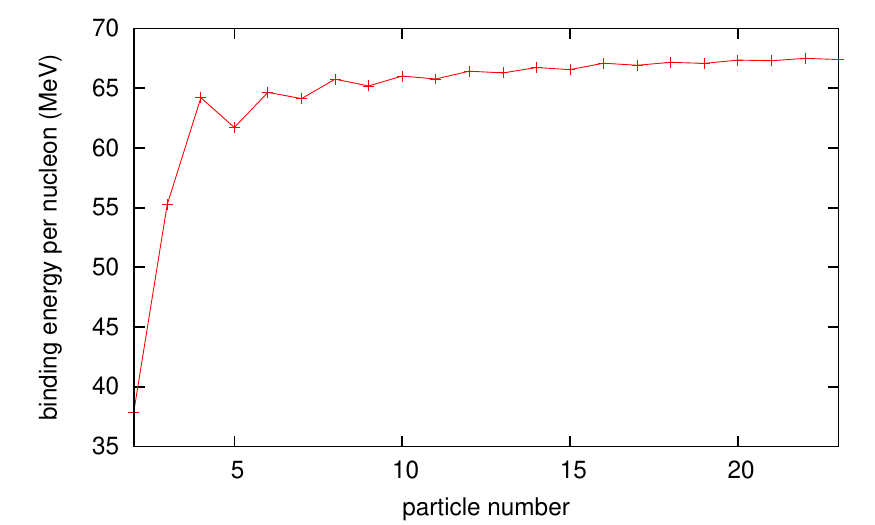}
\end{center}
\caption{Binding energies calculated using rigid body quantization.}
\label{fig:quantum BE}
\end{figure}

In most cases the configuration with the lowest quantum energy is the same as the configuration with the lowest classical energy.  There are two exceptions to this trend.  The quantum corrections to the 6-particle minimizers are relatively large and their order is reversed, so that 6c is the lightest and 6a the heaviest.  The two configurations 10a and 10b have almost identical classical energies, and after quantization the order of their energies is reversed.

The table also lists spin and isospin quantum numbers of the lightest nucleus for each mass number.  In 12 out of 22 cases there is a quantized point particle configuration with the same quantum numbers.  Sometimes the configuration with the correct quantum numbers is not that with lowest energy: for example, 5b has the same spin and isospin as ${}^5\mathrm{He}_2$, but its energy exceeds that of 5a.

Including the quantum corrections gives new predictions for nuclear masses and binding energies, which are plotted in figure \ref{fig:quantum BE}.  The mass of a quantized configuration of $B$ particles is $BM+V+\Delta$, where $V+\Delta$ is the quantum energy recorded in \ref{tab:energies}.  The binding energy is the difference between this quantity and $B$ times the quantized mass of 1 particle.  The calculation of the quantum mechanical correction to the mass of one particle is a standard calculation similar to those described above; the end result is 
\[ M_1 = M + \frac{3\hbar^2}{8L}. \]
The binding energy per nucleon is therefore
\[ \frac{1}{B}\left( B\left(M + \frac{3\hbar^2}{8L}\right) - (BM+V+\Delta)\right) = \frac{3\hbar^2}{8L} - \frac{V+\Delta}{B}. \]
Note that $3\hbar^2/8L \approx 5.521$.

Rigid body quantization has the effect of increasing the binding energy of nuclei, so that they are roughly 5\% of the 1-skyrmion mass rather than 1\%.  It is easy to see why: the quantum correction to the 1-skyrmion mass represents about 5\% of its total mass, whereas the quantum corrections to the masses of larger nuclei represent a much smaller percentage of their total mass.  This means that the quantum corrections to binding energies are also around 5\% of the 1-skyrmion mass, and much larger than the classical binding energies.

The fact that binding energies calculated by rigid body quantization are too large does not represent a failure of the lightly bound Skyrme model, but rather illustrates the pitfalls of rigid body quantization itself.  A collection of $B$ point particles has $6B-3$ degrees of freedom, but in rigid body quantization at most 6 of these are quantized.  Only in the case $B=1$ are all degrees of freedom quantized, so rigid body quantization systematically underestimates the mass of configurations with a large number of particles.  From the point of view of the lightly bound Skyrme model, the degrees of freedom corresponding to moving particles are almost massless, because 1-skyrmions interact only weakly, so arguably these are of comparable importance to the massless degrees of freedom studied in rigid body quantization.

\section{Concluding remarks}
\label{sec:pontificate}\news

We have constructed a simple point particle model of lightly bound skyrmions which almost flawlessly reproduces the results of 
numerical field theoretic energy minimization for charges 1 to 8.\, The only exception is charge 6.\,  Here, the point particle model predicts minimizers with shapes, in order of ascending energy, octahedron, bowtie and pyramid-plus-one, whereas full field simulations find that the correct order is bowtie, octohedron, pyramid-plus-one, albeit with the first two of these very close to degenerate. Alongside this minor blemish one should set some unexpected successes: the point particle model predicted previously unknown energy minimizers at charges 5, 7 and 8, all of which  corresponded to local energy minimizers of the field theory with correct energy ordering. This includes the (so far) lowest energy skyrmion at charge 7.\, The point particle model makes a simple prediction for the inertia tensors of lightly bound skyrmions which, with only two free parameters, fits the field theoretic data for the global minimizers with $1\leq B\leq 8$ to within 
10\%. In judging this, one should bear in mind that an inertia tensor is not a single number, but rather (after accounting for symmetries) 15 independent numbers, so we are actually fitting 120 independent quantities here. 

Having checked consistency with field theory simulations for $1\leq B\leq 8$, we then proceeded to generate local energy minimizers of the point particle model for $9\leq B\leq 23$, where full field simulations are, so far, unavailable. We found that the number of 
nearly degenerate local energy minimizers grows rapidly with $B$, that minimizers consistently resemble subsets of a face centred cubic lattice, with internal orientations correlated with lattice position, and that minimizers often have one fewer than  the maximum possible number of nearest-neighbour bonds. 
We have, furthermore, implemented a simple rigid body quantization scheme for all the local minima we found ($1\leq B\leq 23$). 
As part of this, we devised an automated algorithm to compute the spin-isospin symmetry group of an oriented point cloud, which simultaneously computes the Finkelstein-Rubinstein constraint associated with each symmetry. This allowed us to compute the spin and isospin of the quantum ground states associated (in rigid body quantization) with each local energy minimizer. Since classical binding energies are so small in the lightly bound model, quantization occasionally altered the energy ordering of local minima with a given charge $B$. For 12 baryon numbers (out of 23), this simple quantization procedure produced states corresponding to the spin-isospin data of the lightest nucleus of that baryon number. 

Our numerical scheme to find energy minimizers of the point particle model had two steps: first a crystal-building algorithm was run to generate a subset of the FCC lattice with sufficiently many (nearest neighbour) bonds. Given such a subset, an initial point particle configuration was constructed with particles at the occupied vertices, their internal orientations being fixed by a simple colouring rule, the lattice length scale being chosen to minimize total interaction energy. The second step was to relax this initial FCC subset using a simulated annealing algorithm which allowed the particle positions, and internal orientations, to vary continuously. The results suggest that, in retrospect, this second step is actually superfluous, since the relaxed configuration always stays very close to some FCC lattice
(see figure \ref{fig:wwtp?}). If one merely wishes to find good approximations to classical energy minimizers of the lightly bound Skyrme model, it would seem that considering only FCC lattice subsets, with the lattice scale left as a free parameter, is a fast and effective strategy. It is possible that low energy FCC subsets may also provide useful sets of initial data for energy minimization in more standard variants of the Skyrme model. Certainly this is a quick and convenient means to generate rather uniform initial data of a qualitatively new kind, not obtainable from rational map or alpha particle clustering methods. 

A more interesting problem is to find a better quantization scheme than rigid body quantization. In principle, one could attempt to solve the full Sch\"odinger equation on $\tilde{\mathcal{C}}_B$, subject to the FR constraints. For $B=2$, there is sufficient symmetry that this may well be tractable. For larger $B$, however, it is clearly hopeless. Instead, one should attempt to implement some form of ``vibrational'' quantization scheme, as used for the conventional Skyrme model in \cite{hal,halkinman}. This requires one to find a low-dimensional moduli space of configurations, including all relevant local energy minima, but also configurations interpolating between them, which captures the most important vibrational processes of the classical skyrmion. In this regard, the full point particle model, with positions and orientations allowed to leave the set of FCC configurations, will be essential. Indeed  the model of point skyrmions introduced here may well prove to be an ideal testing ground for vibrational quantization techniques. 

\bigskip
\noindent\textbf{Acknowledgements}  We are grateful to Paul Sutcliffe for helpful conversations. The field theory simulations were performed using code originally developed in collaboration with Juha J\"aykk\"a.

\end{document}